\documentclass[a4paper,USenglish,hyperref,cleveref,autoref,thm-restate,numberwithinsect,nameinlink,pagebackref]{lipics-v2021}
\hideLIPIcs %
\nolinenumbers

\usepackage[utf8]{inputenc}
\usepackage[T1]{fontenc}
\usepackage{amsmath, amsthm, amssymb}
\usepackage{mathtools}
\usepackage{tabularx}
\usepackage{booktabs}
\usepackage{multirow}
\usepackage{dsfont}
\usepackage[]{todonotes} %
\usepackage[numbers,sort]{natbib}
\usepackage{nicefrac}

\usepackage[nolist, nohyperlinks]{acronym}

\usepackage{tikz}
\usetikzlibrary{matrix,fit,decorations.pathreplacing,shapes,patterns,patterns.meta,intersections,positioning,chains,calc,backgrounds,arrows.meta,math,decorations.markings}

\tikzstyle{mycircle}=[circle,draw=black,fill=black!25,fill opacity = 0.3,text opacity=1,inner sep=0pt,minimum size=18pt,font=\small]
\tikzstyle{smallcircle}=[circle, draw=black, fill=white, inner sep=0pt, minimum size=15pt, font=\tiny]
\tikzstyle{vertex_small}=[circle,draw,fill, inner sep=2pt]
\tikzset{
	position/.style args={#1:#2 from #3}{
		at=(#3.#1), anchor=#1+180, shift=(#1:#2)
	}
}

\definecolor{italyGreen}{RGB}{0, 146, 70}
\definecolor{italyRed}{RGB}{206, 43, 55}

\newcommand{\N}{\mathds{N}}
\newcommand{\Z}{\mathds{Z}}
\newcommand{\Q}{\mathds{Q}}
\newcommand{\HI}{\mathds{H}}
\newcommand{\Qnn}{\mathds{Q}^+_0} %
\newcommand{\Qnp}{\mathds{Q}^-_0} %

\acrodef{bdvd}[BDVD]{\textsc{Bounded-Density Vertex Deletion}}

\acrodef{bded}[BDED]{\textsc{Bounded-Density Edge Deletion}}
\acrodef{tbded}[$\tau$-BDED]{$\tau$-\textsc{Bounded-Density Edge Deletion}}
\acrodef{xlc}[X$\ell$C]{\textsc{Exact Cover By $\ell$-Sets}}
\acrodef{vc}[VC]{\textsc{Vertex Cover}}
\acrodef{gf}[GF]{\textsc{General Factors}}
\acrodef{rflow}[\textsc{2GT}]{\textsc{2-Layer General Transshipment}}
\acrodef{flow}[\textsc{Flow}]{\textsc{Maximum $s$-$t$-Flow}}
\acrodef{pm}[\textsc{PM}]{\textsc{Perfect Matching}}

\newcommand{\orient}{\ensuremath{\phi}}
\DeclareMathOperator{\val}{val}
\DeclareMathOperator{\capacity}{cap}

\DeclareMathOperator{\abs}{abs}
\DeclareMathOperator{\tw}{tw}

\newcommand{\problemDef}[3]{%
\begin{center}
	\setlength{\tabcolsep}{2pt}
	\begin{tabular}{@{}lp{12cm}@{}}
		\multicolumn{2}{@{}l}{\textsc{#1}} \\%
		\textbf{Input:} & #2 \\%
		\textbf{Question:} & #3 \\%
	\end{tabular}
\end{center}%
}

\newcommand{\OptproblemDef}[3]{%
\begin{center}
	\setlength{\tabcolsep}{2pt}
	\begin{tabular}{@{}lp{12cm}@{}}
		\multicolumn{2}{@{}l}{\textsc{#1}} \\%
		\textbf{Input:} & #2 \\%
		\textbf{Task:} & #3 \\%
	\end{tabular}
\end{center}%
}

\usepackage{etoolbox}

\newcommand{\appsymb}{$\star$}
\newcommand{\appref}[1]{{\hyperref[proof:#1]{\appsymb}}}

\newcommand{\appendixproof}[2]{%
  \gappto{\appendixProofs}
  {
    \subsection{Proof of \cref{#1}}\label{proof:#1}
    #2
  }
}

\title{Density Matters: A Complexity Dichotomy of Deleting Edges to Bound Subgraph Density}
\titlerunning{Density Matters: A Complexity Dichotomy of $\tau$-Bounded-Density Edge Deletion}

\author{Matthias Bentert}{Technische Universität Berlin, Germany \& University of Bergen, Norway}{bentert@tu-berlin.de}{https://orcid.org/0009-0009-0705-972X}{}

\author{Tom-Lukas Breitkopf}{Algorithmics and Computational Complexity, Technische Universität Berlin, Germany}{t.breitkopf@tu-berlin.de}{https://orcid.org/0009-0008-2875-1945}{}

\author{Vincent Froese}{Algorithmics and Computational Complexity, Technische Universität Berlin, Germany}{vincent.froese@tu-berlin.de}{https://orcid.org/0000-0002-8499-0130}{}

\author{Anton Herrmann}{Algorithmics and Computational Complexity, Technische Universität Berlin, Germany}{a.herrmann@tu-berlin.de}{https://orcid.org/0009-0008-8473-9043}{}

\author{André Nichterlein}{Algorithmics and Computational Complexity, Technische Universität Berlin, Germany}{andre.nichterlein@tu-berlin.de}{https://orcid.org/0000-0001-7451-9401}{}

\authorrunning{M.\ Bentert, T.\ Breitkopf, V.\ Froese, A.\ Herrmann, A.\ Nichterlein}

\Copyright{Matthias Bentert, Tom-Lukas Breitkopf, Vincent Froese, Anton Herrmann, André Nichterlein}

\funding{The research leading to these results has received funding from the European Research Council (ERC) under the European Union’s Horizon 2020 research and innovation program (grant agreement No. 819416) and  the German Research Foundation (DFG, grant 1-5003036: SPP 2378 -- ReNO-2).}


\begin{document}
\maketitle

\begin{abstract}
	We study \ac{tbded}, where given an undirected graph~$G$, the task is to remove as few edges as possible to obtain a graph~$G'$ where no subgraph of~$G'$ has density more than~$\tau$.
	The density of a (sub)graph is the number of edges divided by the number of vertices.
	This problem was recently introduced and shown to be NP-hard for~$\tau \in \{2/3, 3/4, 1 + 1/25\}$, but polynomial-time solvable for~$\tau \in \{0,1/2,1\}$ [Bazgan et al.,~JCSS 2025].
	We provide a complete dichotomy with respect to the target density~$\tau$:
	\begin{enumerate}
		\item If $2\tau \in \N$ (half-integral target density) or~$\tau < 2/3$, then \ac{tbded} is polynomial-time solvable.
		\item Otherwise, \ac{tbded} is NP-hard.
	\end{enumerate}
	We complement the NP-hardness with fixed-parameter tractability with respect to the treewidth of~$G$.
	Moreover, for integral target density~$\tau \in \N$, we show \ac{tbded} to be solvable in randomized~$O(m^{1 + o(1)})$~time.
	Our algorithmic results are based on a reduction to a new general flow problem on restricted networks that, depending on~$\tau$, can be solved via \acl{flow} or \acl{gf}.
	We believe this connection between these variants of flow and matching to be of independent interest.
\end{abstract}

\newpage

\section{Introduction}
\label{sec:intro}

Finding a densest subgraph in a given undirected graph is an ``evergreen research topic''~\cite{LMFB24}.
The density of a (sub)graph is the number of edges divided by the number of vertices, which is equal to half the average degree.
A graph with $n$ vertices can thus have a density up to~$(n-1)/2$; with this maximum being reached if and only if the graph is a clique.
While a reader unfamiliar with the problem might suspect it to be NP-hard---like \textsc{Clique}---it can in fact be solved in polynomial time:
This was first discovered 1979 and published 1982 by \citet{PQ82} with an algorithm based on maximum-flow computations (cf.~\citet{LMFB24}).
An improved algorithm due to~\citet{Gol84} also relies on a maximum s-t-flow. %

In this work, we study the following closely related problem called \acf{tbded} where~$\tau \ge 0$:
Given an undirected graph~$G=(V,E)$, delete as few edges as possible such that the densest subgraph of the remaining graph has density at most~$\tau$.
This problem was recently introduced by \citet{BNV25}, who observed that for~$\tau < 1$ the remaining graph has to be a forest where each tree can have at most~$\lfloor 1 / (1-\tau) \rfloor$ vertices.
Thus, \ac{tbded} generalizes the polynomial-time solvable \textsc{Maximum Cardinality Matching} ($\tau = 1/2$) and the NP-hard \textsc{Perfect $P_3$-Packing} ($\tau = 2/3$).
For~$\tau > 1$, not much is known about the complexity of \ac{tbded} beyond NP-hardness for~$\tau$ slightly above~$1$ ($\tau = 1 + 1/c$ with~$c$ being any constant larger than~$24$~\cite{BNV25}).

Our main result is a complete complexity dichotomy for \ac{tbded}: 
\begin{theorem}
	\label{thm:dichotomy}
	If~$2\tau \in \N$ or~$\tau < 2/3$, then \acl{tbded} is polynomial-time solvable, otherwise it is NP-hard.
	Moreover, if~$\tau \in \N$, then it is solvable in randomized $O(m^{1 + o(1)})$ time.
\end{theorem}

Additionally, we show fixed-parameter tractability for the parameter treewidth. 
To this end, note that for a graph of treewidth~$\tw$ a tree decomposition of width~$\tw$ can be computed in time~$2^{O(\tw^2)}n^{O(1)}$~\cite{Kor23} and a tree decomposition of width~$2\tw + 1$ in time~$2^{O(\tw)}n$~\cite{Kor21}.
\begin{theorem}
	\label{thm:fpt-tw}
	\acl{tbded} is fixed-parameter tractable with respect to treewidth.
	Given a tree decomposition of width~$\tw$ with~$\tau = \frac{a}{b}$, it can be solved in time~${(\max\{a,b\}+1)^{\tw + 1}n^{O(1)}}$.
\end{theorem}

On the technical side, we follow the line of \citet{PQ82} and \citet{Gol84}:
We achieve our algorithmic results by efficiently reducing to a generalized flow problem (which we call \ac{rflow}).
For integral target density~$\tau \in \N$, \ac{rflow} can be solved with a simple maximum $s$-$t$-flow computation in a linear-size network with small integer capacities (at most~$\tau$).
This allows employing recent algorithms for computing maximum flows~\cite{CKLPGS23} to solve \ac{tbded} with~$\tau \in \N$ in randomized $O(m^{1 + o(1)})$ time on graphs with~$m$ edges.
For~$\tau = 1/2$, our flow approach to solve the \ac{rflow} instance fails.
The reason is that we rely on a maximum flow being integral and some arcs having capacity~$\tau$.
In such cases, a simple scaling of capacities yields integer values.
However, our flow-problem then gets an additional constraint that some arcs either transport a full flow (equal to the capacity) or no flow at all.
These ``all-or-nothing'' restrictions render the \textsc{Maximum Flow} problem NP-hard in general, even if all arcs have integer capacities of at most two.\footnote{While we did not find this statement in the scientific literature, an NP-hardness proof was given by Neal Young on stackexchange: \url{https://cstheory.stackexchange.com/questions/51243/maximum-flow-with-parity-requirement-on-certain-edges}. The full proof is included in \cref{app:np_hardness_subset_even_flow}.}
The restricted network structure (being bipartite with all arcs directed towards one partition) in our flow problem, however, allows a reduction to \acl{gf}, which is a generalization of \textsc{Maximum Cardinality Matching}.
Note that this reverses the standard approach of solving \textsc{Maximum Cardinality Matching} in bipartite graphs by a reduction to \textsc{Maximum Flow}~\cite{Sch03}.
Leveraging known algorithmic results for \acl{gf}, we obtain a polynomial-time algorithm for half-integral target density~$\tau$ and fixed-parameter tractability with respect to treewidth. %

\subparagraph{Related Work.} 
We refer to the survey of \citet{LMFB24} for an overview over \textsc{Densest Subgraph} and various approaches to solve it and its variants.

To the best of our knowledge, so far only \citet{BNV25} investigated \ac{bded}, the problem variant of \ac{tbded} where the target density is part of the input and can thus depend on~$n$.
\citeauthor{BNV25} used the allowed dependency on~$n$ to prove W[1]-hardness of \ac{bded} with respect to treewidth~\cite{BNV25} (the proof requires a target density~$\tau = 1 - 1/o(n)$).
Considering the relation between the choice of~$\tau$ and the computational complexity of the problem, they obtained the following results:
\ac{bded} is NP-hard in case~$\tau \in [\frac{2}{3}, \frac{3}{4}]$ or~$\tau = 1 + \frac{1}{c}$, where~$c$ is any constant larger than~$24$, but \ac{bded} becomes polynomial-time solvable in case~$\tau < \frac{2}{3}$ or~$\tau \in [1 - \frac{1}{n}, 1]$.

The vertex deletion variant of \ac{bded} is called \ac{bdvd} and generalizes \textsc{Vertex Cover} ($\tau = 0$) and \textsc{Feedback Vertex Set} ($\tau = 1 - 1/n$).
It was also introduced by \citet{BNV25}, who showed the following results:
\ac{bdvd} is NP-hard for any constant target density.
Both \ac{bded} and \ac{bdvd} are fixed-parameter tractable with respect to the vertex cover number, but W[1]- respectively W[2]-hard with respect to the solution size, that is, the number of edges, respectively vertices, to remove.

\citet{CCK25} studied the approximability of \ac{bdvd}.
Interestingly, they (also) focus on integral target density~$\tau \ge 2$ and provide a polynomial-time $O(\log n)$-approximation that cannot be improved to a factor of~$o(\log n)$ assuming P${}\ne{}$NP.
Their approximation algorithm indeed solves more general problems like \textsc{Supermodular Density Deletion} and \textsc{Matroid Feedback Vertex Set}.

\subparagraph{Organization.} 
After introducing basic notions in \cref{sec:prelim}, we provide an easy flow-construction for integral target density~$\tau$ as warm-up in \cref{sec:integral}.
We show our main result on the generalized flow variant in \cref{sec:flow} and complement our findings with NP-hardness results for \ac{tbded} in \cref{sec:hardness} for the remaining target densities.
We conclude in \cref{sec:concl}.

\section{Preliminaries}
\label{sec:prelim}
We set~$\N := \{0,1,\ldots\}$ to include~$0$ and denote by~$\Qnn := \{q \in \Q \mid q \ge 0\}$ ($\Qnp := \{q \in \Q \mid q \le 0\}$) the set of non-negative (non-positive) rational numbers.
For~$n \in \N$, we define~$[n] \coloneqq \{1,2,\ldots,n\}$ with~$[0] := \emptyset$.
We use $\HI$ to denote the set of all non-negative half-integral rational numbers, that is,~$\HI = \{q \in \Q \mid 2q \in \N\}$.
\subparagraph*{Graph \& Digraphs.}\label{sec:Notation}
All graphs in this work are simple and unweighted.
Let $G$ be an undirected graph and~$D$ be a digraph. 
We denote the set of vertices of $G$ and~$D$ by~$V(G)$ and~$V(D)$ and the set of edges of $G$ and arcs of~$D$ by~$E(G)$ and~$A(D)$, respectively.
For~$C \in \{G,D\}$, we denote with~$n_C$ and~$m_C$ the number of vertices and edges (arcs) of~$C$.
We denote the degree of a vertex~$v\in V(G)$ by~$\deg_G(v)\coloneqq |N(v)|$, where~$N(v)\coloneqq\{u\in V(G)\mid \{u,v\}\in E(G)\}$.
Similarly, we denote the in-degree (out-degree) of a vertex~$v \in V(D)$ by~$\deg_D^-(v)\coloneqq |N^-(v)|$ ($\deg_D^+(v)\coloneqq |N^+(v)|$), where~$N^-(v)\coloneqq\{u\in V(D)\mid (u,v)\in A(D)\}$ ($N^+(v)\coloneqq\{u\in V(D)\mid (v,u)\in A(D)\}$).
If the considered graph is clear from the context, then we drop the subscripts.
We write~$H \subseteq G$ to denote that~$H$ is a subgraph of~$G$.
For a subset~$E'\subseteq E(G)$, we write~$G-E'$ for the subgraph of~$G$ obtained by deleting the edges in~$E'$.
We denote by $K_n$ the complete graph on $n$ vertices (also called a clique of order $n$).

The \emph{density of~$G$} is~$\rho(G)\coloneqq m/n$, where the density of the empty graph is defined as zero.
We denote by~$\rho^*(G)$ the density of the densest subgraph of $G$, that is, $\rho^*(G)\coloneqq \max_{H \subseteq G}\rho(H)$.
A graph $G$ is \emph{balanced} if $\rho(H)\leq \rho(G)$ for every subgraph $H \subseteq G$; this implies~$\rho^*(G) = \rho(G)$.
We study the following problem, which is defined for every rational~$\tau \geq 0$:

\problemDef{\acf{tbded}}
{An undirected graph~$G$ and an integer~$k\in \N$.}
{Is there a subset~$F\subseteq E(G)$ with~$|F|\leq k$ such that~$\rho^*(G- F)\leq\tau$?}

Throughout the rest of this work we assume~$G$ to be connected and hence have~$n_G - 1 \leq m_G$ and~$O(m_G + n_G) = O(m_G)$.
Clearly, if~$G$ is not connected we can obtain a solution by solving each connected component individually and isolated vertices can be ignored.

\subparagraph*{Fractional Orientations.}
A crucial tool in our subsequent proofs and approaches are orientations, which follow from a dual LP formulation by~\citet[Section 3.1]{Cha00} to compute the density of a densest subgraph.
An \emph{orientation} of an undirected graph~$G$ assigns a direction to each edge~$\{u,v\}$, that is,~$\{u,v\}$ is replaced either by~$(u,v)$ or by~$(v,u)$.
\citet{Cha00} observed that the maximum indegree in any orientation of~$G$ is an upper bound for~$\rho^*(G)$.
This is also true for fractional orientations where we assign some part of an edge~$\{u,v\}$ to~$u$ and the rest to~$v$.
More formally, a \emph{fractional orientation}~$\orient$ assigns to each edge~$e = \{u,v\}$ two non-negative rationals~$\orient(e)^u$ and~$\orient(e)^v$ with~$\orient(e)^u + \orient(e)^v = 1$.
We denote with

\begin{align*}
	\deg_{\orient}^-(v) \coloneqq \sum_{u \in N(v)} \orient(\{u,v\})^v && \text{and} && \Delta_{\orient}^- \coloneqq \max_{v \in V(G)} \deg_{\orient}^-(v)
\end{align*}
the \emph{indegree} of a vertex~$v\in V(G)$ and the \emph{maximum indegree} of~$G$ with respect to the fractional orientation~$\phi$.
From now on we write~$\orient(uv)^v$ as a short-hand for~$\orient(\{u,v\})^v$.

We restate the aforementioned property as a lemma together with a stricter formulation.
\begin{lemma}
 \label{lem:density_bound}
 For every graph~$G$ and fractional orientation~$\orient$, it holds that~$\rho^*(G)\le\Delta^-_{\orient}$.
\end{lemma}

\begin{lemma}[\citet{Cha00}]
	\label{lem:density_orientation}
	For every graph $G$, there is a fractional orientation~$\orient$ with $\rho^*(G) = \Delta_{\orient}^-$.
\end{lemma}

It follows from the work of \citet{Cha00} and the fact that LPs can be solved in polynomial time~\cite{Sch03}, that an orientation~$\orient$ with $\rho^*(G) = \Delta_{\orient}^-$ can be computed in polynomial time.

\subparagraph*{Flows and Factors.}
Studied for well over half a century \textsc{Maximum Flow} is still subject of active research~\cite{CKLPGS23}.
We refer to standard text books for an overview on flows~\cite{AMO93,Sch03} and subsequently recall very briefly only the most relevant notions and concepts.
In the basic setting, we are given a directed graph~$D$ forming the flow network with two distinguished vertices~$s,t \in V(D)$ and arcs that are equipped with capacities~$\capacity\colon A(D) \rightarrow\Qnn$.
A flow~$f\colon A(D) \rightarrow\Qnn$ assigns to each arc~$a \in A(D)$ a non-negative amount of flow~${f(a) \leq \capacity(a)}$ such that for each vertex~$v \in V(D) \setminus \{s,t\}$ the amount~$f_{\Sigma}^-(v) := \sum_{u \in N^-(v)} f(u, v)$ of flow entering~$v$ is equal to the amount~$f_{\Sigma}^+(v) := \sum_{u \in N^+(v)} f(v,u)$ of flow leaving~$v$.
Here, and throughout the paper, we write~$f(a, b)$ as a short-hand for~$f((a, b))$.
The excess of the flow at~$v$ is~$f_\Sigma(v) := f_{\Sigma}^+(v) - f_{\Sigma}^-(v)$, that is, $s$ and~$t$ are the only vertices with non-zero excess.
The value of a flow is defined as $\val(f) = f_\Sigma(s) = -f_\Sigma(t)$.

Many extensions on the basic \textsc{Maximum Flow} have been studied, including costs on arcs, multiple sources and sinks, or lower bounds on the flow amount for the arcs (to name just a few).
Computing maximum flows respecting all these restrictions can be done in polynomial time~\cite{AMO93,Sch03}.
Computing a maximum flow that is even on \emph{all} arcs (while still respecting capacities) is still doable in polynomial time (and seems a common exercise given the online search results).
However, restricting a subset~$F \subsetneq A(D)$ of arcs to have either flow~$0$ or~$2$ is already sufficient make the problem NP-hard (see \cref{app:np_hardness_subset_even_flow}).

In contrast to \textsc{maximum flow}, the following \textsc{Matching} generalization is still polynomial-time solvable if restricted to even (or odd) numbers.

\problemDef{\acl{gf}}
{An undirected graph~$G$, a collection of sets~$\mathcal{B} = \{B_v \subseteq \{0, 1, \dots, \deg_G(v)\} \mid v \in V(G)\}$.}
{Is there a subset~$F \subseteq E(G)$, such that~$\deg_{(V(G),F)}(v) \in B_v$ for every vertex~$v \in V(G)$?}

We call a solution~$F$ a \emph{factor} of~$(G, \mathcal{B})$.
\acl{gf} with~$B_v = \{1\}$ for each vertex~$v$ is exactly \acl{pm}.
\citet{Cor88} showed that \acl{gf} is polynomial-time solvable if the maximum gap in any~$B_v$ is of length at most one.
Here, an integer subset~$B \subseteq \Z$ has a gap of length~$d \ge 1$ if there exists an integer~$p \in B$ with~$p+1,\ldots,p+d \notin B$ but~$p+d+1 \in B$.
For example, the set~$\{0,2,3,5\}$ has gaps of length one as certified by~$p = 0$ and~$p=3$.
We denote by \emph{max-gap} the maximum gap length over all sets~$B_v \in \mathcal{B}$.
Note that a max-gap of one is the boundary of polynomial-time solvability: 
There is a straightforward reduction from the NP-hard \textsc{Exact Cover by 3-Sets}~\cite[problem SP2]{GJ79} to \acl{gf} where every vertex~$v$ has allowed degrees either~$B_v = \{0,3\}$ (the vertex represents a set) or~$B_v=\{1\}$ (the vertex represents an element)~\cite{Cor88}.
Clearly, the constructed instance has max-gap two.

Recently, \citet{MSS21} provided an algorithm for \acl{gf} running in~${(M+1)^{3 \tw}}$ time on graphs of treewidth~$\tw$, where~$M$ is the largest integer appearing in any set~$B_v$.

\section{Warm Up: Integral Target Density}
\label{sec:integral}
Consider a \acl{tbded} instance~$\mathcal{I} = (G, k)$ with integral target density~$\tau$.
We transform~$\mathcal{I}$ into a \textsc{Maximum Flow} instance with~$O(m_G)$ vertices and arcs and small integral capacities.
As a maximum flow in such networks can be computed in randomized~$O(m^{1 + o(1)})$ time~\cite{CKLPGS23}, we get a near-linear time algorithm for \ac{tbded} if~$\tau \in \N$.

The flow shall correspond to a fractional orientation.
That is, one unit of flow should leave each edge~$\{u,v\}$; a part of it flowing to~$u$, the rest to~$v$.
Moreover, at each vertex~$v$ a total flow of at most~$\tau$ is allowed to arrive.
This idea easily translates to the following construction of the flow network~$D$ (see \cref{fig:bded_to_flow} for a visualization of the construction):%
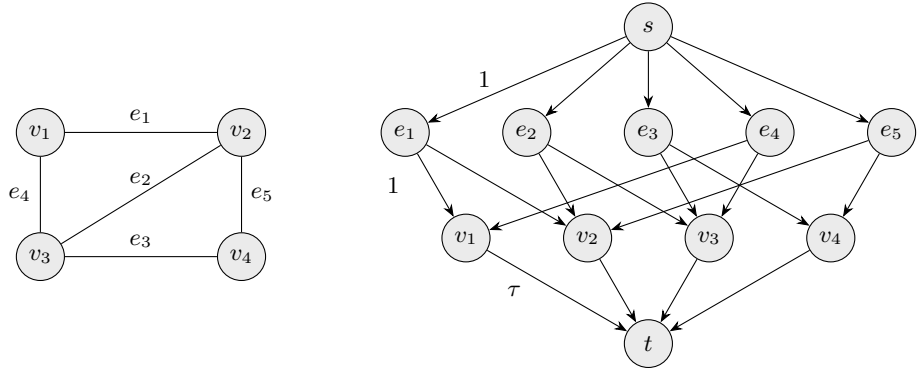
\begin{figure}
	\centering
	\begin{tikzpicture}[myarrow/.style={-Stealth},node distance=1cm and 1cm]
			\begin{scope}
				\begin{scope}[node distance=1cm and 2cm]
					\node[mycircle] (v1) {$v_1$};
					\node[mycircle, right=of v1] (v2) {$v_2$};
					\node[mycircle, below=of v1] (v3) {$v_3$};
					\node[mycircle, below=of v2] (v4) {$v_4$};
				\end{scope}

				\foreach \i/\j/\txt/\p in {%
					v1/v2/$e_1$/above,
					v2/v3/$e_2$/above,
					v3/v4/$e_3$/above,
					v1/v3/$e_4$/left,
					v2/v4/$e_5$/right}
					\draw [-] (\i) -- node[font=\small,\p] {\txt} (\j);
			\end{scope}

			\begin{scope}[xshift=8cm,yshift=1.4cm]
				\begin{scope}[yscale=1.4,xscale=1.6]
					\node[mycircle] (s) at (0,0) {$s$};
					\foreach \i in {1,...,5}
					{
						\node[mycircle] (e\i) at (-3 + \i, -1) {$e_\i$};
					}
					\foreach \i in {1,...,4}
					{
						\node[mycircle] (v\i) at (-2.5 + \i, -2) {$v_\i$};
					}
					\node[mycircle] (t) at (0,-3) {$t$};
				\end{scope}

				\foreach \i/\j in {s/e2, s/e3, s/e4, s/e5, e1/v2, e2/v2, e2/v3, e3/v3, e3/v4, e4/v1, e4/v3, e5/v2, e5/v4, v2/t, v3/t, v4/t}
					\draw [myarrow] (\i) -- (\j); 

				\foreach \i/\j/\txt/\p in {s/e1/$1$/left, e1/v1/$1$/left, v1/t/$\tau$/left}
					\draw [myarrow] (\i) -- node[font=\small,\p, xshift=-10pt] {\txt} (\j); 

			\end{scope}
		\end{tikzpicture}
	\caption{Example transformation of a \ac{tbded} instance with~$\tau \in\N$ (on the left) to a flow instance (on the right).
	All arcs on the same layer of the flow network have the same capacity as indicated at the leftmost arc.}
	\label{fig:bded_to_flow}
\end{figure}
\begin{itemize}
 \item Set the vertex set of~$D$ as $V(D) = V(G) \cup E(G) \cup \{s, t\}$. 
	This means for every vertex in~$G$ as well as every edge in~$G$, there is a vertex in~$D$.
 \item For every $e = \{u,v\} \in E(G)$ add an arc with capacity~1 to~$D$ going from source~$s$ to vertex~$e$.
 \item For every~$v \in V(G)$ add an arc with capacity~$\tau$ to~$D$ going from~$v$ to sink~$t$.
 \item For every edge~$e = \{u,v\} \in E(G)$ add two arcs with capacity~1 to~$D$, one going from vertex~$e$ to vertex~$u$ and one going from~$e$ to~$v$.
\end{itemize}
Consider a maximum $s$-$t$-flow~$f$ in the constructed network~$D$.
Since all capacities are integral, it follows that~$f$ is integral~\cite{AMO93,Sch03}.
It is not hard to see that we have~$\val(f) = m_G$ (that is, all arcs leaving~$s$ transport a flow of~$1$) if and only if there is a fractional orientation~$\phi$ in~$G$ with~$\Delta_{\orient}^- \le \tau$; the flow on the arcs in~$E(G) \times V(G)$ corresponds to~$\phi$.
However, a stronger statement holds: 
There is a solution~$F \subseteq E(G)$ for~$\mathcal{I}$ of size~$|F| = k$ if and only if~$f$ has value~$m_G - k$ .
The correspondence between~$f$ and~$F$ becomes clear at the arcs incident to~$s$: For an edge~$e \in E(G)$ we have~$f(s,e)=0$ if and only if~$e \in F$.
We defer the formal proof of this statement to the next section (see \cref{lem:bded_iff_flow}), where we provide a transformation to a more general flow-variant that covers also non-integral target densities.

Note that the flow network~$D$ has size~$O(m_G + n_G) = O(m_G)$ and can be constructed in linear time.
Moreover, the largest reasonable value for $\tau$ is~$\binom{n_G}{2}$.
Using fast existing maximum-flow algorithms for polynomially bounded integral capacities~\cite{CKLPGS23}, we obtain the following:

\begin{theorem}
	There is a randomized algorithm solving \acl{tbded} with integral target density~$\tau\in \N$ in time $O(m^{1 + o(1)})$.
\end{theorem}

A further consequence of the flow-construction is that if $\tau \in \N$, then there exists a fractional orientation that is in fact integral.
That is, each edge is assigned to either of its two endpoints.
This was previously observed by \citet{CCK25}.

Note that for $\tau = 1/2$ the capacities in~$D$ are no longer integral. 
Therefore, an integral maximum flow is not guaranteed.
However, \ac{tbded} for $\tau = 1/2$ corresponds to \textsc{Maximum Cardinality Matching}.
Hence, for our particular network~$D$ there is a polynomial-time algorithm that computes among all flows that are integral on the arcs leaving~$s$ a maximum one.
Note that a linear-time reduction from \textsc{Maximum Cardinality Matching} to \textsc{Flow} is not known (so far).
Thus, a polynomial-time algorithm for \ac{tbded} with half-integral capacities cannot rely on such a simple reduction as provided above.
We present the solution to this issue in the next section.

\section{Algorithmic Results}
\label{sec:flow}
In this section, we provide efficient algorithms for \ac{tbded} with certain target densities.
More precisely, we provide polynomial-time algorithms for integral and half-integral target densities, while for arbitrary constant target densities, we prove fixed-parameter tractability with respect to the treewidth of the input graph.
Our approach builds on the overall idea of \cref{sec:integral}, that is, we reduce \ac{tbded} to some kind of flow problem.
In particular, we introduce the \acf{rflow} problem, which is, informally spoken, a flow problem on graphs with multiple source and sink vertices, where every edge goes directly from a source to a sink and every vertex has a set of allowed excess values.
Afterwards, we translate this problem to \acl{gf}.
Applying known results for \acl{gf}, we finally deduce algorithmic results for \ac{rflow}, which, on the one hand, are of interest on their own and, on the other hand, can then be translated to \ac{tbded}.

\subsection{Transshipment Formulation}
\label{ssec:transshipment}
To describe \ac{tbded} more generally, we switch to transshipment problems, that is, a flow-variant where we have multiple sources and sinks with specified amounts of flow that are provided resp.\ consumed by the sources and sinks~\cite{Sch03}.
Consider a digraph~$D$.
For a function~$b\colon V(D) \rightarrow \Q$, a $b$-transshipment is a flow~$f$ with~$f_\Sigma(v) = b(v)$ for all vertices~$v \in V(D)$.
Note that $\sum_{v\in V}b(v) = 0$ is necessary for a $b$-transshipment to exist.
Adapting notation from \acl{gf}, we specify a set collection~$\mathcal B = \{B_v \subseteq \Q\mid v\in V(D)\}$ of allowed excess values for each vertex.
A \emph{$\mathcal B$-transshipment} is then a (flow) function~$f\colon A(D) \to \Qnn$ with~$f_\Sigma(v) \in B_v$ for all~$v \in V(D)$.
Note that if also arc capacities~$\capacity\colon A(D) \to \Qnn$ are provided, then we have the usual additional requirement~$f(a) \le \capacity(a)$ for each arc~$a \in A(D)$.
We call an instance~$(D, \capacity, \mathcal B)$ \emph{integral} if~$\capacity\colon A(D) \to \N$ and~$B_v \subseteq \Z$ for all~$v \in V(D)$.
We call a digraph~$D$ \emph{2-layered}, if the underlying undirected graph is bipartite with parts~$U$ and~$W$ and all arcs go from~$U$ to~$W$, that is, $A(D) \subseteq U \times W$.
For a 2-layered network~$D$ with vertex partitions~$U$ and~$W$ and flow~$f$, we define the \emph{value}~$\val(f) = \sum_{u \in U} f_{\Sigma}(u) = \sum_{w \in W} -f_{\Sigma}(w)$.
Note that in a 2-layered transshipment instance, we can assume that the vertices in~$U$ are sources (non-negative excess) and the vertices in~$W$ are sinks (non-positive excess).
With these notions, we can define our transshipment problem:

\OptproblemDef{\acf{rflow}}
{A 2-layered network~$D$ with parts~$U$ and~$W$, arc capacities~$\capacity\colon A(D) \rightarrow \Qnn$, and a collection $\mathcal{B} = \{B_v \subseteq \Qnn \mid v\in U\}\cup\{B_v \subseteq \Qnp \mid v\in W\}$ of sets.}
{Find a $\mathcal B$-transshipment~$f$ in~$D$ which maximizes $\val(f)$.}

For the rest of this work we assume the underlying undirected graph of~$D$ to be connected and therefore have~$n_D - 1 \leq m_D$ and~$O(m_D + n_D) = O(m_D)$.
Clearly if this is not the case, we obtain a solution by solving each connected components individually and isolated vertices can be ignored.
To familiarize ourselves with the problem, let us model an instance~$\mathcal{I}= (G, k)$ of \ac{tbded} as a \ac{rflow}-instance; we denote this linear-time transformation by~$T(\mathcal{I})$:
\begin{itemize}
	\item The vertex set of the network~$D$ is~$V(D) \coloneqq V(G) \cup E(G)$.
	\item For each edge~$e = \{u,v\} \in E(G)$, add the arcs~$(e,u)$ and~$(e,v)$ with capacity~1 to~$D$.
	\item For each~$e \in E(G)$, set~$B_e \coloneqq \{0,1\}$, that is, a flow of either 0 or 1 is allowed to leave~$e$.
	\item For each~$v \in V(G)$, set~$B_v \coloneqq [-\tau,0]$, that is, any flow up to~$\tau$ is allowed to arrive in~$v$.
\end{itemize}

Note that in the above construction a vertex~$v \in V(G)$ is assigned an infinitely large set~$[-\tau,0]$ of allowed excess values. 
We henceforth assume that a reasonable encoding for the sets~$B_v$ is used (here, e.\,g.\ via the numbers~$-\tau$ and~$0$ encoding the interval).

It is straightforward to verify the correctness of this reduction:

\begin{lemma}
	\label{lem:bded_iff_flow}
	Let~$\mathcal{I}= (G, k)$ be a \ac{tbded} instance and let~$T(\mathcal{I}) = (D,\capacity, \mathcal{B})$ be the constructed \ac{rflow} instance.
	Then there is a $\mathcal B$-transshipment~$f$ in~$D$ with value~$\val(f) \geq m_G - k$ if and only if there is an edge set~$F \subseteq E$, $|F| \leq k$, with~$\rho^*(G-F) \le \tau$.
\end{lemma}

\begin{proof}
	``$\Rightarrow$'': Suppose that~$f$ is a $\mathcal B$-transshipment for~$(D, \capacity, \mathcal{B})$ with~$\val (f) \geq m_G -k$.
	We define~$F \subseteq E(G)$ to be the set of edges where no flow is sent from the corresponding vertices in~$D$, that is, $F:= \{e \in E(G) \mid f_{\Sigma}(e) = 0 \}$.
	Since~$f_\Sigma(e)\in B_e = \{0,1\}$ for all~$e\in E(G)$ and~$\val(f) \geq m_G - k$, it follows that~$|F| \leq k$.

	Using the flow~$f$, we now construct a fractional orientation~$\orient$ for~$G - F$ with~$\Delta_{\orient}^- \leq \tau$.
	\Cref{lem:density_bound} then yields the claim.
	For any edge~$e = \{u,v\} \in E(G) \setminus F$, consider the arcs~$(e, u)$ and~$(e, v)$ in~$A(D)$.
	Set~$\orient(e)^u := f(e, u)$ and~$\orient(e)^v := f(e, v)$.
	It holds that~$\orient(e)^u + \orient(e)^v = 1$, since $f_{\Sigma}(e) = 1$.
	For any vertex~$v \in V(G)$, it holds that~$\deg_{\orient}^-(v) \leq \tau$, because~$f_{\Sigma}(v) \geq -\tau$ by the excess constraint~$B_v = [-\tau, 0]$.

	``$\Leftarrow$'': Now, let~$F \subseteq E(G)$, $|F| \leq k$, with~$\rho^*(G-F) \leq \tau$. By~\Cref{lem:density_orientation}, there exists a fractional orientation~$\orient$ on $G-F$ with~$\Delta_{\orient}^- \leq \tau$.
	We construct a $\mathcal B$-transshipment~$f$ with~$\val(f)\ge m_G - k$ from~$\orient$:
	For each~$e = \{u,v\} \in E(G) \setminus F$, set the flows~$f(e, u) \coloneqq \orient(e)^u$ and~$f(e, v) \coloneqq \orient(e)^v$, and for each~$e = \{u,v\} \in F$, set $f(e, u) = f(e, v) \coloneqq 0$.
	Consider the vertices in~$V(D)$ representing the edges of~$E(G)$.
	The above assignment ensures for any vertex~$e \in E(G) \setminus F$, that~$f_{\Sigma}(e) = 1$ and for any vertex~$e \in F$ that~$f_{\Sigma}(e) = 0$.
	Moreover, $f_{\Sigma}(v) \in [-\tau, 0]$ is guaranteed since~$\Delta_{\orient}^- \leq \tau$.
	Finally, we have~$\val(f) = \sum_{e \in E\setminus F} 1 = m_G - |F| \geq m_G - k$.
\end{proof}

\subsection{Solving \acl{rflow} via General Factors}
\label{ssec:general_factors}
We now consider how to solve \ac{rflow} using \acl{gf}.
One important step will be to limit the excess constraints of vertices in a \ac{rflow} instance to integral values.
For this, we first consider conditions under which \ac{rflow} yields integral solutions, meaning solutions in which the flow assigned to any edge is integral.

\begin{lemma}
	\label{lem:integral_flow_bgt}
	Let~$(D, \capacity\colon A(D)\to \N,\mathcal{B})$ be a \ac{rflow} instance with~$A(D) \subseteq U \times W$.
	Furthermore, for each~$v \in V(D)$, let~$B_v \subseteq \Z$ or~$B_v = [c_v, d_v]$ for some~${c_v, d_v \in \Z}$ with~$c_v \leq d_v$.
	If~$(D,\capacity, \mathcal{B})$ has a~$\mathcal B$-transshipment, then it attains an integral maximum $\mathcal B$-transshipment.
\end{lemma}

\begin{proof}
	The idea is to construct a flow instance with lower and upper bounds for the flow on each arc, then use a known integrality property and finally transfer the flow back to our \ac{rflow} instance.
	Let~$f$ be a maximum~$\mathcal B$-transshipment for~$(D, \capacity, \mathcal{B})$ (such~$f$ must exist if there is a feasible transshipment) and consider the flow instance~$\mathcal{I}'$, where we take the network~$D$ together with a source vertex~$s$, a sink vertex~$t$ and the following arcs with lower and upper bounds (denoted by intervals) on the flow assigned to them:
	\begin{itemize}
	 \item For each arc~$a \in A(D)$, we set the permitted flow to~$[0,\capacity (a)]$.
	 \item For every~$u \in U$ with~$f_{\Sigma} (u) > 0$, we add the arc~$(s,u)$ and distinguish between two cases. If~$B_u = [c_u,d_u]$, then we set the permitted flow to~$[c_u, d_u]$. Otherwise, we set the permitted flow to~$[f_{\Sigma}(u), f_{\Sigma}(u)]$.
	 \item For every~$w \in W$ with~$f_{\Sigma} (w) < 0$, we add the arc~$(w,t)$ and distinguish between two cases. If~$B_w = [c_w,d_w]$, then we set the permitted flow to~$[-d_w, -c_w]$ and otherwise to~$[-f_{\Sigma}(w), -f_{\Sigma}(w)]$.
	\end{itemize}
	We now show that there is a feasible flow~$f'$ for the constructed instance, which then implies the existence of an integral maximum flow for~$\mathcal{I}'$ (see \cite[Corollary 11.2e]{Sch03}).
	Let~$a$ be some arc in our constructed instance. We define~$f'$ by
	\begin{equation*}
	f'(a) \coloneqq
		\begin{cases}
			f_{\Sigma}(u) & \text{if } a = (s,u) \text{ for some } u \in U \\
			f (a) & \text{if } a \in A(D) \\
			- f_{\Sigma}(w) & \text{if } a = (w,t) \text{ for some } w \in W.
		\end{cases}
	\end{equation*}
	Note that~$\val (f) = \val(f')$ since both flows send the same amount of flow from~$U$ to~$W$.
	The feasibility of~$f'$ follows directly from our construction and the fact that~$f$ is a $\mathcal{B}$-transshipment.
	Now, let~$f''$ be an integral maximum flow for~$\mathcal{I}'$.
	We define an integral $\mathcal B$-transshipment~$f^*$ by~$f^* (u,w) := f''(u,w)$ for all~$(u,w) \in A(D)$.

	First, we show that~$f^*$ is indeed a~$\mathcal B$-transshipment.
	For this, note that for each~$u \in U$, we have~$f^*_{\Sigma}(u) = f''(s,u)$.
	We distinguish between two cases.
	If~$B_u \subseteq \Z$, then the permitted flow on~$(s,u)$ is exactly~$f_{\Sigma}(u)$ and therefore~$f''(s,u) = f_{\Sigma}(u) \in B_u$.
	Otherwise, the permitted flow on~$(s,u)$ is~$[c_u,d_u]$ which directly implies~$f''(s,u) \in [c_u, d_u] = B_u$.
	For all vertices~$w \in W$, we have~$f^*_{\Sigma}(w) = - f''(w,t)$ and by analogous arguments as above, it follows~$-f''(w,t) \in B_w$.

	It remains to show that~$f^*$ is indeed a maximum flow.
	Recall that~$f$ is a maximum flow for~$(D, \capacity, \mathcal{B})$ and~$f'$ with~$\val(f') = \val(f)$ is a feasible solution for~$\mathcal I'$.
	Further,~$f''$ is a maximum integral flow for~$\mathcal I'$, that is, $\val(f') \leq \val(f'')$.
	Finally, since $\val(f^*) = \val(f'')$, we have~$\val (f) = \val (f') \leq \val (f'') = \val (f^*)$.
\end{proof}

Using the above lemma we now observe that any \ac{rflow} instance with an integral maximum transshipment can be restricted to integral capacities and integral excess values.

\begin{observation}
	\label{obs:integral_excess_constraints_bgt}
	If a \ac{rflow} instance~$(D, \capacity, \mathcal{B})$ has an integral maximum~$\mathcal B$-transshipment, then it is equivalent to~$(D, \capacity', \mathcal{B}')$ where~$\mathcal B' \coloneqq \{B_v \cap \Z \mid B_v \in \mathcal{B}\}$ and~$\capacity '(a) \coloneqq \lfloor \capacity (a) \rfloor$ for each arc~$a \in A(D)$.
\end{observation}

With this observation we can now describe a reduction of integral \ac{rflow} instances to \acl{gf}.
This will later allow us to use known algorithms for~\acl{gf} to solve~\ac{rflow} instances.
\begin{lemma}
	\label{lem:bgt_to_gf_reduction}
	Let~$(D, \capacity\colon A(D)\to \N,\mathcal{B} = \{B_v \subseteq \Z \mid v\in V(D)\})$ be an integral~\ac{rflow} instance, where
	all capacities are bounded polynomially in the size of~$D$.
	Then, one can construct in polynomial time, a \acl{gf} instance~$(G, \mathcal{B}')$ with the following properties:
			\begin{enumerate}[(i)]
				\item $(D,\capacity, \mathcal{B})$ has an integral $\mathcal B$-transshipment of value~$p$ if and only if~$(G, \mathcal{B}')$ has a factor of size~$2p$. \label{prop:correctness}
				\item Given a tree decomposition for the underlying undirected graph of~$D$ of width~$\tw$, we can compute a tree decomposition of width~$\max\{2,\tw\}$ for~$G$ in polynomial time.
				\label{prop:treewidth}
				\item The max-gap of~$\mathcal{B}'$ is~$\max\{1,\gamma\}$, where~$\gamma$ is the max-gap of~$\mathcal{B}$. \label{prop:maxgap}
			\end{enumerate}
\end{lemma}

\begin{proof}
	Let~$(D,\capacity, \mathcal{B} = \{B_v \subseteq \Z \mid v\in V(D)\})$ be a given~\ac{rflow} instance with~${A(D) \subseteq U \times W}$.
	Now, consider the multigraph, which is obtained by replacing every arc~$(u,w)$ with~$\capacity (u,w)$ many undirected parallel edges between~$u$ and~$w$.
	Moreover, set~${\mathcal{B}' \coloneqq \{ \abs(B_v) \mid B_v \in \mathcal{B} \}}$, where~$\abs(B_v) \coloneqq \{|x| \mid x \in B_v\}$.
	We can transform the multigraph into a simple graph by subdividing every edge and adding the degree constraint~$\{0,2\}$ to~$\mathcal{B}'$ for every newly created vertex.
	We refer to the set of new vertices by~$H$ and if necessary, we call the vertex, which was introduced for the $i^{\text{th}}$ parallel edge between~$u$ and~$w$, by~$h^i_{uw}$.
	We call the obtained graph~$G$ and the constructed~\acl{gf} instance is~$(G, \mathcal{B}')$.
	We next show property~(\ref{prop:correctness}).
	
	``$\Rightarrow$'': Let~$f$ be an integral~$\mathcal B$-transshipment for~$(D, \capacity, \mathcal{B})$.
	We show that there is a general factor~$F \subseteq E(G)$ for the instance~$(G, \mathcal{B}')$ of size~$|F| = 2 \val (f)$.
	Intuitively, if the flow value on the arc~$(u,w)$ is~$f(u,w)$, then we add~$2f(u,w)$ edges from~$G$ to~$F$, such that they correspond to~$f(u,w)$ copies of the edge~$\{u,w\}$.
	Formally, we set
	\[F := \bigcup_{(uw) \in A(D)} \bigcup_{i = 1}^{f(u,w)} \left\{ \{u,h^i_{uw}\}, \{h^i_{uw},w\}\right\}.\]
	Since every flow from~$U$ to~$W$ counts towards~$\val (f)$, it directly follows that~$|F| = 2 \val (f)$.
	It remains to show that~$F$ satisfies the degree constraints~$\mathcal{B}'$, that is, $\deg_{(V(G), F)} (v) \in B_v$ for every~$v \in V(G)$.
	Recall, that for~$h \in H$, we have~$B_h = \{0,2\}$.
	By construction, we either add zero or two edges incident to~$h$ to~$F$ and therefore we have~$\deg_{(V(G), F)} (h) \in B_h$.
	For~$v \in U \cup W$, note that the degree in~$(V(G),F)$ coincides exactly with the absolute value of~$f_{\Sigma}(v)$.
	This together with~$f_{\Sigma} (v) \in B_v$ implies that~$\deg_{(V(G),F)} (v) \in \abs (B_v) \in \mathcal{B}'$.

	``$\Leftarrow$'': Now suppose~$F \subseteq E(G)$ is a $\mathcal B'$-factor.
	We show that there exists a $\mathcal B$-transshipment~$f$ with $\val (f) = \frac{|F|}{2}$.
	Let~$(u,w)$ be an arc in~$D$.
	By construction, if~$\{u,h^i_{uw}\} \in F$ for some~${h^i_{uw} \in H}$, then also~$\{h^i_{uw},w\} \in F$, because~$B_{h^i_{uw}} = \{0,2\}$.
	Now, we define a flow~$f$ by setting~$f(u,w)$ to the number of vertices from~$H$ to which both~$u$ and~$w$ are connected in~$(V(G),F)$.
	This is an integral feasible flow since for every arc~$(u,w) \in A(D)$ both~$u$ and~$w$ have by construction at most~$\capacity (u,w)$ common neighbors in $G$.
	Moreover, for each~$u\in U$, we have~${f_\Sigma(u) = \deg_{(V(G),F)}(u)\in \abs(B_u)=B_u}$ and for each~$w\in W$, we have~${f_\Sigma(w)= -\deg_{(V(G),F)}(w)\in \{-x\mid x\in \abs(B_w)\}=B_w}$.
	Since exactly~$\frac{|F|}{2}$ edges are between~$U$ and~$H$ and also between~$H$ and~$W$, it follows~$\val(f)=\frac{|F|}{2}$.

	It remains to show the two properties (\ref{prop:treewidth}) and (\ref{prop:maxgap}).

	(\ref{prop:treewidth}): Let~$(\mathcal T, (X_i)_{i \in V(\mathcal T)})$ be a tree-decomposition of width~$\tw$ for the underlying undirected graph of~$D$.
	For each arc~$(u,w) \in A(D)$, there exists a bag~$X_i$ with~$\{u,w\}\subseteq X_i$.
	We obtain a tree decomposition for~$G$ of width~$\max\{2,\tw\}$ by appending one bag~$\{u,w,h^j_{uw}\}$ to~$X_i$ for each~$j \in [\capacity (u,w)]$.

	(\ref{prop:maxgap}): Note that, for any vertex~$h \in H$, the set~$B_h$ has a gap of exactly one.
	For~$v \in U \cup W$, the collection~$\mathcal{B}'$ contains the set~$\abs (B_v)$ which has the same gap as~$B_v$ since either~$B_v \subseteq \Qnn$ or~$B_v \subseteq \Qnp$.
\end{proof}

To sum up, we have seen that a \ac{rflow} instance allowing flows in an interval with integer endpoints can be restricted to an integral instance and that integral \ac{rflow} instances can be translated into equivalent \acl{gf} instances.
For non-integral \ac{rflow} instances, the reduction to \acl{gf} does not work.
A common trick we will use later, is to scale capacities and excess constraints by some factor in order to make them integral.
This scaling obviously influences the overall flow value.

\begin{lemma}
	\label{lem:scaling_bgt}
	Let $(D, \capacity, \mathcal{B})$ be a \ac{rflow}~instance, $q \in \Qnn$ and~$(D, \capacity ', \mathcal{B}')$ be the instance, which is obtained by multiplying all excess constraints and capacities from~$(D, \capacity, \mathcal{B})$ with~$q$.
	Then,~$(D, \capacity, \mathcal{B})$ has a maximum~$\mathcal B$-transshipment value of~$p$ if and only if~$(D, \capacity ', \mathcal{B}')$ has a maximum $\mathcal B'$-transshipment value of~$pq$.
\end{lemma}
\begin{proof}
  Clearly, since all capacities and excess values are multiplied by~$q$, it follows that~$f$ is a maximum~$\mathcal B$-transshipment if and only if~$f'$
  with $f'(u,w)\coloneqq q f(u,w)$ for all~$(u,w)\in A(D)$ is a maximum~$\mathcal B'$-transshipment.
\end{proof}

We write~$q\mathcal{B}$ for~$\mathcal{B}'$ to indicate that~$\mathcal{B}'$ was obtained from~$\mathcal{B}$ by scaling with factor~$q$.

\subsection{Algorithmic Consequences}
\label{ssec:algorithmic_consequences}

\subparagraph{2-Layer General Transshipment.}
Having discussed how to translate a \ac{rflow}~instance into a \acl{gf}~instance, we now consider the algorithmic implications for the \ac{rflow} problem. We obtain three algorithmic results as stated in the following theorem.

\newpage
\begin{theorem}
	\label{thm:bgt_algorithms}
	Let~$(D, \capacity\colon A(D)\to \N , \mathcal{B})$ be an instance of \ac{rflow}.
	\begin{enumerate}[(i)]
		\item If~$B_v = [a_v,b_v] \cap \Z$ for each vertex~$v$ with~$a_v \le 0 \le b_v$ and the excess constraints and capacities are polynomially bounded in the size of~$D$, then one can compute in randomized~$O(m^{1 + o(1)})$ time a maximum~$\mathcal B$-transshipment.
		\label{thm:bgt_with_s-t-flow}
		\item If~$B_v \subseteq \Z$ has gap at most one for each vertex~$v$ and capacities are bounded polynomially, then one can compute in polynomial time a maximum $\mathcal B$-transshipment or output that there exist no feasible flow. \label{thm:bgt_gap_one}
		\item Suppose we are given a tree decomposition for the underlying undirected graph of~$D$ of width~$\tw$. Moreover, suppose that~$B_v \subseteq \Z$ for each vertex~$v$ and that capacities are bounded polynomially.
		Let~$M\coloneqq \max_{v \in V(D)} \max \{|x|\mid x\in B_v\}$.
		Then, one can compute in~$(M+1)^{\tw + 1} n^{O(1)}$ time a maximum $\mathcal B$-transshipment or output that there exists no feasible flow. \label{thm:bgt_treewidth}
	\end{enumerate}
\end{theorem}

\begin{proof}
	We prove each of the three statements separately.

	(\ref{thm:bgt_with_s-t-flow}): We provide a linear-time reduction to~\textsc{Maximum Flow}. If the edge capacities are integers and polynomially bounded, then an (integral) maximum flow for this instance can be computed in randomized $m^{1 + o(1)}$ time \cite{CKLPGS23}.
	Now consider the following \textsc{Maximum Flow} instance, where we take the network~$D$, add a source vertex~$s$, a sink vertex~$t$ and the following arcs:
	\begin{itemize}
	\item For every~$u \in U$, we add the arc~$(s,u)$ with capacity~$b_u$.
	\item For every~$w \in W$, we add the arc~$(w,t)$ with capacity~$-a_w$.
	\end{itemize}
	This network can be constructed in~$O(m_D + n_D) = O(m_D)$ time.
	We now compute a maximum flow~$f$ for this network and obtain a~$\mathcal B$-transshipment by setting~$f^* (u,w) \coloneqq f(u,w)$ for all~$(u,w) \in A(D)$.
	Note that~$f^*$ is feasible since~$f^*_\Sigma(u)=f^+_{\Sigma} (u) \in [0, b_u]$ for all~$u\in U$ and~$f^*_\Sigma(w)=-f^-_{\Sigma} (w) \in [a_w,0]$ for all~$w\in W$.
	Furthermore, $f^*$ is a maximum flow since~$\val (f^*) = \val (f)$ and any flow in the transshipment instance can be directly translated to an~$s$-$t$-flow~$f'$ of the same value as follows:
	\begin{equation*}
	f'(a) \coloneqq
		\begin{cases}
			f^*_{\Sigma}(u) & \text{if } a = (s,u) \text{ for some } u \in U \\
			f^* (a) & \text{if } a \in A(D) \\
			- f^*_{\Sigma}(w) & \text{if } a = (w,t) \text{ for some } w \in W.
		\end{cases}
	\end{equation*}

	(\ref{thm:bgt_gap_one}): We apply \cref{lem:bgt_to_gf_reduction} and obtain a \acl{gf}~instance~$(G, \mathcal{B}')$ which has gap exactly one by~\cref{lem:bgt_to_gf_reduction}~(iii).
	Using an algorithm from \citet{DP18} we can now find a maximum $\mathcal B'$-factor~$F$ of~$G$ in polynomial time.
	In particular, \cref{lem:bgt_to_gf_reduction} gives us a $\mathcal B$-transshipment of value~$\frac{|F|}{2}$ and implies that this is indeed the maximum.
	On the other hand, if there is no solution to~$(G, \mathcal{B}')$, then there is also no feasible flow for~$(D, \capacity, \mathcal{B})$.

	(\ref{thm:bgt_treewidth}): We again apply \cref{lem:bgt_to_gf_reduction} and obtain a \acl{gf}~instance~$(G, \mathcal{B}')$.
	By \cref{lem:bgt_to_gf_reduction}~(\ref{prop:treewidth}), we get a tree decomposition for~$G$ of width at most~$\tw + 1$ and, by construction of the \acl{gf}~instance, we know that~$\max \{\max B' \mid B' \in \mathcal{B}'\}\le M$ (note that we can assume~$M\ge 2$ since otherwise we are in case~(\ref{thm:bgt_gap_one}) which is polynomial-time solvable).
	This allows us to apply a result from \citet{MSS21}, which states that a maximum $\mathcal B'$-factor of~$G$ can be computed in~$(M+1)^{\tw + 1}n^{O(1)}$ time.
	If there is no solution to~$(G, \mathcal{B}')$, then there is also no feasible flow for~$(D, \capacity, \mathcal{B})$.
\end{proof}

\subparagraph{$\tau$-Bounded-Density Edge Deletion.}
In \cref{ssec:transshipment} we described how to reduce \ac{tbded} to \ac{rflow}.
Therefore, the results from \Cref{thm:bgt_algorithms} for \ac{rflow} translate directly to \ac{tbded} and imply the algorithmic results in \cref{thm:dichotomy,thm:fpt-tw}:

\newpage
\begin{theorem}
\label{thm:bded_algorithms}
	Let~$\mathcal{I} = (G, k)$ be a \acl{tbded}~instance.
	\begin{enumerate}[(i)]
		\item If~$\tau \in \N$, then one can decide~$\mathcal I$ in randomized $O(m^{1 + o(1)})$ time.
		\label{thm:bded_integral}
		\item If~$2\tau \in \N$, then one can decide $\mathcal I$ in polynomial time. \label{thm:bded_half_integral}
		\item If $\tau = \frac{a}{b}$ for some fixed~$a,b \in \N$ and a tree decomposition for~$G$ of width~$\tw$ is given, then one can decide~$\mathcal I$ in~$(\max\{a,b\}+1)^{\tw + 1}n^{O(1)}$ time. \label{thm:bded_treewidth}
	\end{enumerate}
\end{theorem}

\begin{proof}
	Recall the construction of a \ac{rflow}~instance~$T(\mathcal{I})=(D, \capacity, \mathcal{B})$ from \Cref{ssec:transshipment}, which has maximum $\mathcal B$-transshipment of value at least~$m_G - k$ if and only if~$(G, k, \tau)$ is a yes-instance (see \cref{lem:bded_iff_flow}).
	The capacities on the arcs in~$D$ are one and the excess constraints are of the form~$\{0, 1\}$ and $[-\tau, 0]$. 
	Since~$\tau$ is bounded polynomially in~$n_D$, clearly so are all the capacities and excess constraints. 
	This also holds when scaling the instance with a constant, as described in \cref{lem:scaling_bgt}.
	Using the construction, we can derive the three statements.

	(\ref{thm:bded_integral}): First, note that $T(\mathcal{I})$ can be constructed in~$O(m_G + n_G) = O(m_G)$ time.
	The excess constraints are of the form~$\{0,1\}$ and~$[-\tau,0]$.
	Since~$\tau \in \N$ and a feasible solution exists (e.g. sending a flow of~0 along all arcs), we can apply \cref{lem:integral_flow_bgt} which guarantees us the existence of an integral maximum~$\mathcal B$-transshipment.
	Subsequently, \cref{obs:integral_excess_constraints_bgt} tells us that we can restrict all the excess constraints to their intersection with~$\Z$.
	This allows us to apply \Cref{thm:bgt_algorithms}~(\ref{thm:bgt_with_s-t-flow}), which concludes the algorithm.

	(\ref{thm:bded_half_integral}): By \cref{lem:scaling_bgt}, we can scale the \ac{rflow} instance~$T(\mathcal{I})$ by the factor two and the obtained excess constraints are of the form~$\{0,2\}$ and~$[-2\tau,0]$.
	Since~$2\tau$ is integer and a feasible solution exists, we can again apply \cref{lem:integral_flow_bgt} which gives us the existence of an integral maximum flow.
	Again, \cref{obs:integral_excess_constraints_bgt} tells us that we can restrict all the excess constraints to their intersection with~$\Z$.
	Note that the excess constraints are now of the form~$\{0,2\}$ and~$[-2\tau, 0] \cap \Z$.
	In particular, they have gap at most one and therefore the obtained transshipment instance satisfies the conditions of \Cref{thm:bgt_algorithms}~(\ref{thm:bgt_gap_one}).
	Since we scaled~$T(\mathcal I)$ by two, the maximum~$2\mathcal B$-transshipment value is at least~$2(m_G - k)$ if and only if~$(G, k, \tau)$ is a yes-instance.

	(\ref{thm:bded_treewidth}): By \cref{lem:scaling_bgt}, we can scale the \ac{rflow} instance~$T(\mathcal{I})$ by the factor~$b$ and the obtained excess constraints are of the form~$\{0,b\}$ and~$[-b\tau,0] = [-a,0]$.
	This scaled version has a maximum $b\mathcal B$-transshipment of value at least~$b(m_G - k)$ if and only if the original instance has a maximum flow of at least~$m_G - k$.
	By the same arguments as above, we can restrict the excess constraints to their intersections with~$\Z$.
	Furthermore, we easily obtain a tree decomposition for the underlying undirected graph of~$D$ of width~$\tw$ since the graph can be obtained by subdividing all edges of~$G$ \cite{CyganFKLMPPS15}.
	To finish the proof, note that we can now apply~\cref{thm:bgt_algorithms}~(\ref{thm:bgt_treewidth}) to compute a maximum transshipment in~$(\max\{a,b\}+1)^{\tw + 1}n^{O(1)}$ time.
	We then return yes if and only if the value is at least~$b(m_G - k)$.
\end{proof}

\section{NP-hardness Results}
\label{sec:hardness}
In this section, we complement our algorithmic findings from \cref{sec:flow} by showing NP-hardness for every constant target
density~$\tau \in \{q \in \Q \mid q \ge \nicefrac{2}{3} \land q\notin \HI\}$, thereby completing the proof of \cref{thm:dichotomy}.
We prove the result in two steps.
First, we show that any fixed target density~$\tau \in [\nicefrac{2}{3},1)$ yields NP-hardness.
Afterwards, we show that \ac{tbded} is NP-hard for any constant target density~$\tau \in \{q \in \Q \mid q>1 \land q \notin \HI\}$.

\subparagraph*{Case~$\tau < 1$.}
Note that for $\tau < 1$ any solution graph (the graph obtained by removing at most~$k$ edges) cannot contain a cycle as a cycle has density~$1$.
Hence, any solution graph is a forest.
The maximum density of a forest is determined by the size of a largest tree in it as a tree with~$n$ nodes has~$n-1$ edges and thus density~$1-\nicefrac{1}{n}$, which is monotonically increasing for increasing values of~$n$.
Thus, any target density~$\tau <1$ can be replaced by the largest value~$\tau' \leq \tau$ with~$\tau'=\nicefrac{\ell}{\ell+1}$ for some~$\ell \in \N$.
For~$\ell=0$, we have to remove all edges and thus \ac{tbded} becomes trivial.
For~$\ell=1$, the problem is equivalent to finding a smallest set of edges that only leaves a matching.
This can be solved in polynomial time by computing a maximum matching and returning all non-matching edges.
For~$\ell\in \{2,3\}$, \citet{BNV25} showed NP-hardness.
We extend their hardness result for~$\ell=3$ to arbitrary~$\ell$.
The proof consists of a reduction from \acf{xlc} and is sketched in \cref{fig:hardness_below_one}.
\begin{figure}
	\centering
        \begin{tikzpicture}[scale=0.80,node distance=8mm]%
			\def\xcoor{2}
			\def\c{4}

			\begin{scope}[xshift=-2cm,start chain=going right]
				\foreach \i in {1,2}{
					\node[on chain,circle,draw, fill, inner sep=2pt] (x-\i)[label=below:{$y_\i$}] {};
				}
				\node[on chain, circle, inner sep=2pt] (x-dots1) {\ldots};
				\node[on chain, circle, draw,fill,inner sep=2pt, label=below:{$y_{\ell}$}] (x-l) {};
				\node[on chain, circle, inner sep=2pt] (x-dots2) {\ldots};
				\node[on chain, circle, draw,fill,inner sep=2pt, label=below:{$y_{\ell q}$}] (x-n) {};
			\end{scope}

			\begin{scope}[yshift=3cm,start chain=going right]
				\foreach \i in {1,2}{
					\node[on chain, circle,draw,fill,inner sep=2pt,label=above:{$c_\i$}] (c-\i){};
				}
				\node[on chain, circle,inner sep=2pt] (c-dots) {\ldots};
				\node[on chain, circle,draw,fill,inner sep=2pt, label=above:{$c_t$}] (c-m) {};
			\end{scope}

			\begin{scope}[yshift=3cm,xshift=7cm,start chain=going right]
				\node[on chain, circle,draw,fill,inner sep=2pt,label=above:{$u_1$}](u-1) {};
				\node[on chain, circle,draw,fill,inner sep=2pt,label=above:{$u_2$}, xshift=1cm](u-2) {};

				\node[on chain, circle,inner sep=2pt] (u-dots) {\ldots};
				\node[on chain, circle,draw,fill,inner sep=2pt, label=above:{$u_{t-q}$}] (u-m) {};
			\end{scope}

			\begin{scope}
                \def\circSize{5}
                \def\circEdgeLen{1cm}

                \tikzmath{\angleRange=90;
                        \offset = \angleRange/(\circSize-1);
                        \start=90+((360-\angleRange)/2)-\offset;}
                \foreach \c in {u-1, u-2, u-m}
                {
                    \foreach \i in {1,...,\circSize}
                    {
                        \node[position=\start+\offset*\i:{\circEdgeLen} from \c](tmp){};
                        \draw[gray] (\c) -- (tmp);
                    }
                }

                \node[vertex_small, inner sep=1, position=\start+\offset:{\circEdgeLen} from u-1] {};
                \node[position=\start+\offset:\circEdgeLen*1.02 from u-1] {$z_{1}^{1}$};
                \node[vertex_small, inner sep=1, position=\start+\circSize*\offset:{\circEdgeLen} from u-m] {};
                \node[position=\start+\circSize*\offset:\circEdgeLen*1.02 from u-m] {$z_{t-q}^{\ell-1}$};
			\end{scope}

			\draw[thick](x-1)--(c-1);
			\draw[thick](x-2)--(c-1);
			\draw[thick](x-l)--(c-1);

			\draw[thick](x-2)--(c-2);
			\draw[thick](x-n)--(c-m);
			\draw[thick] (x-n) -- (c-2);

			\draw[color=gray](c-1)--(x-dots1);
			\draw[color = gray](c-1) edge[bend left=16] (x-dots1);
			\draw[color = gray](c-1) edge[bend right=16] (x-dots1);

			\draw[color=gray](x-dots1) -- (c-2);

			\draw[color = gray](x-dots2) -- (c-2);
			\draw[color = gray](x-dots2) edge[bend left=10] (c-2);
			\draw[color = gray](x-dots2) edge[bend left=20] (c-2);

			\draw[color = gray](x-dots2)--(c-m);
			\draw[color = gray](x-dots2) edge[bend left=16] (c-m);
			\draw[color = gray](x-dots2) edge[bend right=16] (c-m);
			\draw[color = gray](x-dots1) edge[bend left=10] (c-m);
			\draw[color = gray](x-dots1) edge[bend right=10] (c-m);

			\node[inner sep = 16pt, very thick, italyGreen, draw,rounded corners, dashed, fit=(c-1) (c-2) (c-m),label=above:$C$](C) {};
			\node[inner sep = 12pt, very thick, italyRed, draw,rounded corners, dashed, fit=(x-1) (x-2) (x-n),label=below:$X$] (X){};
			\node[inner sep = 9pt, very thick, black!40, draw,rounded corners, dashed, fit=(u-1) (u-2) (u-dots) (u-m),label=above:$U$] (U){};

			\draw[gray] ([yshift=.1cm]C.east) -- ([yshift=-.3cm]U.west);
			\draw[gray] ([yshift=-.2cm]C.east) -- ([yshift=-.1cm]U.west);
			\draw[gray] ([yshift=.4cm]C.east) -- ([yshift=.1cm]U.west);
			\draw[gray] ([yshift=.2cm]C.east) -- ([yshift=.3cm]U.west);
			\draw[gray] ([yshift=-.1cm]C.east) -- ([yshift=.4cm]U.west);

			\end{tikzpicture}
	\caption{A schematic illustration of the reduction in \cref{thm:hardness_below_one}.
	Given an instance of \ac{xlc} with universe~$X = \{x_1, x_2, \dots, x_{\ell q}\}$ and set collection~$C = \{C_1, C_2, \dots, C_t\}$, we add a vertex~$y_i$ for each element~$x_i \in X$ and a vertex~$c_j$ for each set~$C_j \in C$.
	We add an edge~$\{y_i, c_j\}$ whenever~$x_i \in C_j$.
	We further add dummy stars with~$\ell-1$ leaves and fully connect their centers to the vertices representing~$C$.
	Any \ac{tbded} solution then leaves~$t$ stars with~$\ell$ leaves and there is a a one-to-one mapping between (green) star centers and sets selected in an \ac{xlc} solution.
	}
	\label{fig:hardness_below_one}
\end{figure}
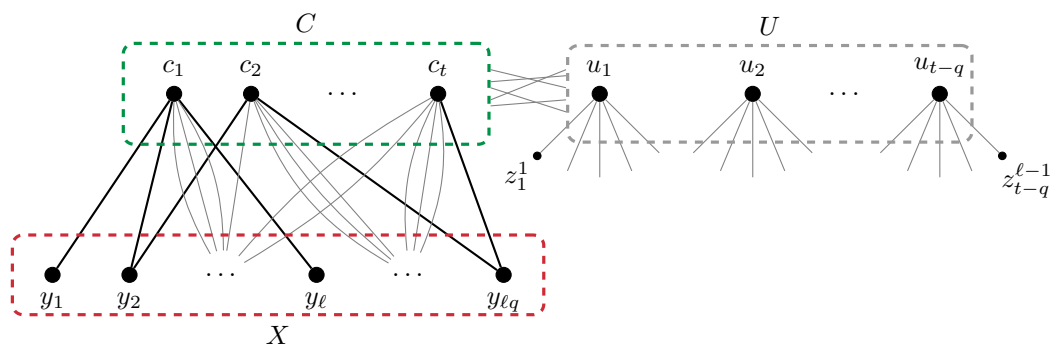

\begin{theorem}
	\label{thm:hardness_below_one}
	\ac{tbded} is NP-hard for any~$\tau = \nicefrac{\ell}{\ell+1}$ with~$\ell \geq 3$.
\end{theorem}
\begin{proof}
	We reduce from \acl{xlc}, which is NP-hard for~$\ell \geq 3$~\cite{KH78}.
	It is defined as follows.

	\problemDef{\acf{xlc}}
	{A set~$X$ with~$|X| = \ell q$ and a collection~$C$ of~$\ell$-element subsets of~$X$.}
	{Is there a subcollection~$C'$ of~$C$ where every element of~$X$ occurs in exactly one set in~$C'$?}

	The proof is a straightforward adaption of the hardness proof for~$\tau=\nicefrac{3}{4}$~\cite{BNV25}; see \cref{fig:hardness_below_one} for a sketch.
	Let~$(X, C)$ be an instance of \ac{xlc} with~$X = \{x_1, x_2, \dots, x_{\ell q}\}$ and~$C = \{C_1, C_2, \dots, C_t\}$.
	If~$t < q$, then we trivially have a no-instance as not all elements appear in a set in~$C$.
	So we assume~$t \geq q$ and construct a graph~$G$ as follows:
	We start with one vertex~$y_i$ for each element~$x_i \in X$ and one vertex~$c_j$ for each set~$C_j \in C$.
	We add an edge~$\{y_i, c_j\}$ whenever~$x_i \in C_j$.
	We then add a set of~$t-q$ stars~$S_1,S_2,\ldots,S_{t-q}$ each with~$\ell-1$ leaves to~$G$.
	Formally, for each~$i \in [t-q]$, we add a vertex~$u_i$ and for each~$i\in [t-q]$ and each~$j \in [\ell-1]$, we add a vertex~$z_i^j$ and the edge~$\{u_i,z_i^j\}$.
	Next, we add all edges~$\{c_j,u_i\}$ with~$j \in [t]$ and~$i \in [t-q]$.
	This concludes the construction of~$G$.
	Let~$m$ be the number of edges in it.
	To conclude the entire construction, we set~$k \coloneqq m - \ell t$.

	Since the construction takes polynomial time, it only remains to show correctness.
	To this end, first assume that there is a solution~$C'$ for the input instance~$(X,C)$ of \ac{xlc}.
	Note that~$|C'|=q$.
	We construct a set~$B$ of~$\ell t$ edges such that deleting all edges except for those in~$B$ is a solution for the constructed instance of \ac{tbded}.
	For each~$C_j \in C'$, include the~$\ell$ edges joining~$c_j$ to vertices~$y_i$ with~$x_i \in C_j$.
  Note that there are exactly~$\ell q$ such edges.
  Moreover, for each~$i \in [t-q]$ and each~$j \in [\ell-1]$, we add the edge~$\{u_i,z_i^j\}$ to~$B$.
  Finally, for each~$C_j \notin C'$, we pick a distinct vertex~$u_i$ and include~$\{u_i, c_j\}$ in~$B$.
  Note that there are exactly~$t-q$ sets in~$C \setminus C'$ and also~$t-q$ vertices~$u_i$. Therefore, $B$ contains a perfect matching between these two sets of vertices.
  Overall, $B$ has size~$\ell q + (\ell-1)(t-q) + (t-q) = \ell t$.
  Note that each vertex~$c_j$ belongs to a different connected component in the graph induced by~$B$.
  Moreover, each such vertex belongs to a tree with~$\ell$ additional vertices (either a star~$S_j$ or~$\ell$ $y$-vertices).
  Thus, the maximum density is~$\nicefrac{\ell}{\ell+1}$, concluding the first part of the proof.

  For the other direction, assume that there is a set~$F$ of at most~$k$ edges whose removal yields a graph with maximum density at most~$\tau=\nicefrac{\ell}{\ell+1}$.
  Note that the constructed graph~$G$ has~$(\ell+1) t$ vertices.
  Moreover, $G - F$ has at least~$m - |F| \geq m - k = \ell t$ edges.
	Since any cycle has density~$1$, $G-F$ is a forest and since any tree with~$\ell+2$ vertices has~$\ell+1$ edges and therefore density~$\nicefrac{\ell+1}{\ell+2} > \nicefrac{\ell}{\ell+1}$, it holds that each tree in~$G-F$ contains at most~$\ell+1$ vertices.
	If~$G-F$ contains exactly~$t$ trees, each with~$\ell+1$ vertices and~$\ell$ edges, then~$G-F$ contains~$\ell t$ edges and~$F$ contains~$m-\ell t = k$ edges.
	Since $G$ is connected (assuming each element appears in at least one set in~$C$), increasing the number of trees in~$G-F$ or decreasing the size of any tree in~$G-F$ increases the size of~$F$ to above~$k$.
	Thus, $G-F$ contains exactly~$t$ trees and each tree contains~$\ell+1$ vertices and~$\ell$ edges.
	This implies that no edge~$\{u_i,z_i^j\}$ for any~$i \in [t-q]$ and any~$j \in [\ell-1]$ is contained in~$F$.
	Moreover, each vertex~$u_j$ is adjacent to exactly~$\ell$ vertices in~$G-F$ and thus exactly one edge~$\{u_i,c_j\}$ is not contained in~$F$ for each~$i \in [t-q]$ (and the index~$j$ is unique for each~$i$).
	Consider the remaining~$q$ vertices~$c_j$ that are not connected to a star~$S_i$.
	We claim that the corresponding sets~$C_j$ are a solution for the original instance of \ac{xlc}.
	Consider an arbitrary vertex~$y_i$ (corresponding to an element~$x_i \in X$).
	By construction, $y_i$ is only adjacent to vertices~$c_j$ with~$x_i \in C_j$.
	Since each tree in~$G-F$ contains exactly~$\ell+1$ vertices and only~$q$ vertices~$c_j$ and~$\ell q$ vertices~$y_i$ are not yet distributed, each such undistributed vertex~$c_j$ must be adjacent to~$\ell$ vertices~$y_i$ each and each vertex~$y_i$ is only adjacent to a single vertex~$c_j$ in~$G-F$.
	Note that this exactly corresponds to an exact cover, concluding the proof.
\end{proof}

\subparagraph*{Case~$\tau > 1$.}
Our goal is to show that \ac{tbded} is NP-hard for any constant~$\tau = \nicefrac{p}{q} > 1$ with~$\tau \notin \HI$.
The proof will be a reduction from \acl{vc} on~$q$-regular graphs which mostly consists of two types of gadgets: \emph{vertex gadgets} and \emph{edge gadgets}.
On a high level, the vertex gadgets are balanced graphs with density~$\tau$.
The edge gadgets have density slightly below~$\tau$, but there are also two edges connecting each edge gadget to the two corresponding vertex gadgets.
Combined, an edge gadget and the two corresponding vertex gadgets have density slightly larger than~$\tau$ (in such a way that~$\nicefrac{(m-\nicefrac{1}{q})}{n}=\tau$).
The idea is that deleting one edge in a vertex gadget is sufficient to ``fix'' all adjacent edge gadgets.

The vertex gadget will be based on the following lemma.
We state it in a more general form which can be applied whenever a solution graph contains a cycle.
It states that in a cycle, we can always assume that at least one edge in a fractional orientation is fully assigned to one of its endpoints.
The idea is to simply shift the fractional orientation around the cycle until this is no longer possible in which case one of the edges fully assigns its fractional orientation to one endpoint.

\begin{lemma}
	\label{lem:1-0-edge}
	Let~$G$ be a graph containing a cycle $C$ and let~$\orient$ be a fractional orientation of~$G$.
	Then, there exists a fractional orientation~$\orient'$ with~$\deg_{\orient'}^-(v) = \deg_{\orient}^-(v)$ for all~$v \in V(G)$ and an edge~$e=\{u,v\}$ in~$C$ with~$\orient'(e)^u = 1$, that is, $e$ is completely assigned to one of its endpoints by the fractional orientation~$\orient'$.
	Given~$G$,~$C$, and~$\orient$, then~$\orient'$ can be found in linear time with respect to the size of~$G$.
\end{lemma}

\begin{proof}
	Let~$C=(v_1, \dots, v_\ell, v_{\ell + 1} = v_1)$.
	If there already exists an edge~$\{v_i,v_{i+1}\}$ in~$C$ with~$\orient(v_iv_{i+1})^{v_i} = 1$ or~$\orient(v_iv_{i+1})^{v_{i+1}} = 1$, then we are done.
	Otherwise, let~${\varepsilon_i := \orient(v_iv_{i+1})^{v_i}}$ for each~$i \in [\ell]$, $\varepsilon := \min_{i \in [\ell]} \varepsilon_i$ and let~$j$ be an index such that~$\varepsilon = \varepsilon_j$.
	We obtain the desired fractional orientation~$\orient'$ from~$\orient$ by setting~$\orient'(e) = \orient(e)$ for all edges in~$G$ that are not~$\{v_i,v_{i+1}\}$ for some~$i\in [\ell]$,~$\orient'(v_iv_{i+1})^{v_{i+1}} \coloneqq \orient(v_iv_{i+1})^{v_{i+1}} + \varepsilon$ and~$\orient'(v_iv_{i+1})^{v_{i}} = \orient(v_iv_{i+1})^{v_{i}} - \varepsilon$ for all~$i \in [\ell]$.
	Note that this is possible by the definition of~$\varepsilon$ and~$\orient'(v_jv_{j+1})^{v_{j+1}} = 1$.
	Moreover, $\deg_{\orient'}^-(v) = \deg_{\orient}^-(v)$ for all~$v \in V(G)$.
	Given~$G$,~$C$ and~$\orient$, clearly~$\orient'$ can be determined in linear time with respect to the size of~$G$ considering edges in the cycle at most twice and all other edges at most once.
\end{proof}

For the edge gadget, the basic idea is to start with a path adding additional edges in regular intervals to construct a near balanced graph with two orientations such that all vertices except for the two endpoints of the path have indegree~$\tau$ and one endpoint has indegree $\tau-1$ and the other has indegree~$\tau-1+\nicefrac{1}{q}$.
Connecting each endpoint of the path via a single edge with the corresponding vertex gadget then ensures that at least one edge in one of the three gadgets has to be deleted.
Recall that we want to encode that a vertex is part of a vertex cover if an edge is removed in the corresponding vertex gadget.
The excess value of~$\nicefrac{1}{q}$ will be small enough to distribute the available capacity generated by a single edge deletion to all incident edge gadgets.

\begin{lemma}%
\label{lem:edge_gadget}
	For every constant rational number~$\tau = \frac{p}{q} > 1$, where~$p$ and~$q$ are coprime, a graph $W$ with two special vertices~$u$ and~$w$ can be constructed in constant time, such that there exist two fractional orientations~$\orient_1$ and~$\orient_2$ with the following properties:
	\begin{enumerate}[(i)]
		\item $\deg_{\orient_1}^-(v) = \deg_{\orient_2}^-(v) = \tau$ for all~$v \in V(W)\setminus\{u,w\}$,
		\item $\deg_{\orient_1}^-(u) = \deg_{\orient_2}^-(w) = \tau - 1 + \frac{1}{q}$ and
		\item $\deg_{\orient_1}^-(w) = \deg_{\orient_2}^-(u) = \tau - 1$.
	\end{enumerate}
\end{lemma}
{
\begin{proof}
	Let~$p = cq + r$ where~$c$ and~$r$ are integers and~$0 \leq r < q$.
	Note that~$r \geq 1$ since~$p$ and~$q$ are coprime and~$c \geq 1$ since~$\tau > 1$.
	Let~$p'\coloneqq p-(c-1)q=q+r$, and~$\tau' \coloneqq \nicefrac{p'}{q}$.
	Note that~$1<\tau'<2$ and that~$p'$ and~$q$ are coprime.
	We construct the graph~$W$ in two stages.
	First, we construct a graph~$W'$ for target density~$\tau'$ and then we use this construction to construct the graph~$W$ for target density~$\tau$.

	Note that if the graph~$W'$ has~$n$ vertices and fulfills the requirements for the fractional orientations (for the target density~$\tau'$), then the number~$m$ of edges in it fulfills~${m = \tau' n - 2 + \frac{1}{q}}$.
	This number is not necessarily an integer, but since~$p'$ and~$q$ are coprime, we can choose~$n$ such that~$(p' \cdot n)\bmod q = q-1$ (see for example~\citet{burton2006elementary} for more details).
	Then, ${m = \tau' n - 2 + \frac{1}{q} = \frac{p'n + 1 - 2q}{q}}$ is an integer.
	Moreover, we can increase~$n$ by~$aq$ for any positive integer~$a$ and this increases~$m$ by~$a p'$ and thus~$m$ remains an integer.
	Since~$2q > p' > q$ as~$1 < \tau' < 2$, this allows us to pick\footnote{Use the extended Euclidean algorithm to first get an integer~$t$, such that~$1 \equiv t p' \mod q$.
	Setting~$n' = ((q-1)t) \bmod q$ satisfies~$(n'p') \bmod q = ((q-1) t p') \bmod q = q-1$.
	Let~$m' = \frac{p'}{q}n' -2 + \frac{1}{q}$.
	If~$n'+1 \leq m'$ (implying~$n' \geq 4$¸), set~$n = n'$, $m = m'$, otherwise iteratively increment~$n'$ by~$q$ until the condition holds.
	Since every increase of~$n'$ by~$q$ increases~$m'$ by~$p'$ with~$p' < q$ and since initially~$n' < q$, this procedure terminates after at most~$q$ steps.}
	a positive integer~$n \geq 4$ such that~$m$ is also a positive integer with~$n+1 \leq m \leq 2n-3$.

	To construct~$W'$ with~$n$ vertices and~$m$ edges, we start with a path~${(u=v_1, \dots, v_n = w)}$ and iteratively add edges to the construction while simultaneously ensuring the existence of the desired fractional orientation~$\orient'_1$.
	We will later show how to modify~$\orient'_1$ to get~$\orient'_2$.
	First, we orient the edge~$\{u,v_2\}$ with a value of~$\tau' - 1 + \nicefrac{1}{q}$ to~$u$ and with a value of~${1 - (\tau' - 1 + \nicefrac{1}{q}) = 2 - \frac{p'-1}{q}}$ to~$v_2$.
	This is possible since~$1 \leq \tau' = \frac{p'}{q} < \frac{p'+1}{q} \leq 2$.
	Now starting from~$i = 2$ and for increasing values, we orient the edge~$\{v_i,v_{i+1}\}$ in such a way, that~$v_i$ gets indegree exactly~$\tau'$ from the already oriented edges until we reach a point where even assigning a value of~$1$ is not enough.
	In the latter case, we add the edge~$\{v_i,w\}$ to the graph and orient it completely to~$v_i$.
	We then continue with the edge~$\{v_i,v_{i+1}\}$ as before but now, whenever a value of~$1$ is not enough, we add the edge~$\{u,v_i\}$ and completely orient it towards~$v_i$.
	See \cref{fig:edge_gadget} for an illustration.
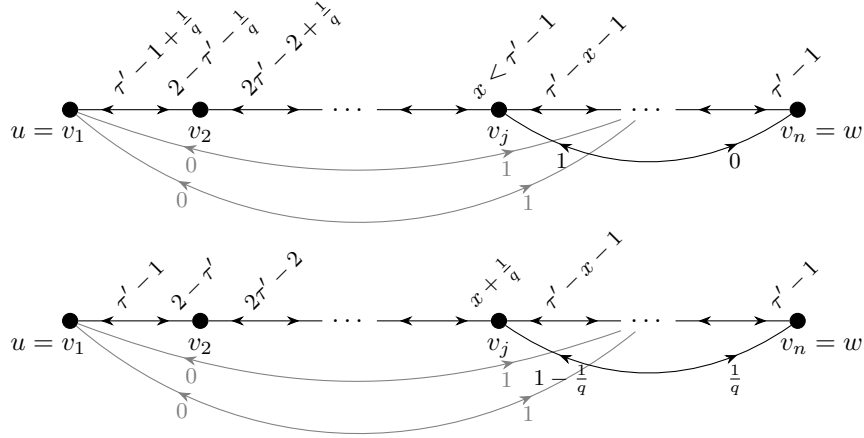
\begin{figure}
	\centering
	\begin{tikzpicture}[node distance = 1.5cm]
        \begin{scope}[decoration={markings, mark=at position 0.2 with{\arrowreversed{Stealth}}, mark=at position 0.8 with{\arrow{Stealth}}}, start chain=going right]
            \node[on chain, vertex_small, label={[label position=below]{$u = v_{1}\phantom{ = u}$}}] (v1){};
            \node[on chain, vertex_small,label=below:{$v_{2}$}] (v2){};
            \node[on chain] (dots1) {\ldots};
            \node[on chain, vertex_small,label=below:{$v_{j}$}] (v3){};
            \node[on chain] (dots2) {\ldots};
            \node[on chain, vertex_small, label=below:{$\phantom{w = }v_{n}=w$}] (vn) {};

            \foreach \i/\j/\toi/\toj in {v1/v2/$\tau'-1+\frac{1}{q}$/$2-\tau'-\frac{1}{q}$, v2/dots1/$2\tau'-2+\frac{1}{q}$/, dots1/v3/ /$x<\tau'-1$, v3/dots2/$\tau'-x-1$/, dots2/vn/ /$\tau'-1$}
            {
                \draw[postaction={decorate}] (\i) -- (\j) node[pos=0.2, above, label={[rotate=45]right:\small\toi}]{} node[pos=0.7, above, label={[rotate=45]right:\small\toj}]{};
            }
            \foreach \i in {1,...,2}
            {
                \draw[gray, postaction={decorate}] [bend right = \i*20]  (v1) to node[pos=0.2, below]{\footnotesize0} node[pos=0.8, below]{\footnotesize1} (dots2);
            }
            \draw[postaction={decorate}] [bend right = 35]  (v3) to node[pos=0.2, below]{\footnotesize1} node[pos=0.8, below]{\footnotesize0} (vn);

        \end{scope}

        \begin{scope}[yshift=-2.8cm, decoration={markings, mark=at position 0.2 with{\arrowreversed{Stealth}}, mark=at position 0.8 with{\arrow{Stealth}}}, start chain=going right]
            \node[on chain, vertex_small, label={[label position=below]{$u = v_{1}\phantom{ = u}$}}] (v1){};
            \node[on chain, vertex_small,label=below:{$v_{2}$}] (v2){};
            \node[on chain] (dots1) {\ldots};
            \node[on chain, vertex_small,label=below:{$v_{j}$}] (v3){};
            \node[on chain] (dots2) {\ldots};
            \node[on chain, vertex_small, label=below:{$\phantom{w = }v_{n}=w$}] (vn) {};

            \foreach \i/\j/\toi/\toj in {v1/v2/$\tau'-1$/$2-\tau'$, v2/dots1/$2\tau'-2$/, dots1/v3/ /$x + \frac{1}{q}$, v3/dots2/$\tau'-x-1$/, dots2/vn/ /$\tau'-1$}
            {
                \draw[postaction={decorate}] (\i) -- (\j) node[pos=0.2, above, label={[rotate=45]right:\footnotesize\toi}]{} node[pos=0.7, above, label={[rotate=45]right:\footnotesize\toj}]{};
            }
            \foreach \i in {1,...,2}
            {
                \draw[gray, postaction={decorate}] [bend right = \i*20]  (v1) to node[pos=0.2, below]{\footnotesize0} node[pos=0.8, below]{\footnotesize1} (dots2);
            }
            \draw[postaction={decorate}] [bend right = 35]  (v3) to node[pos=0.2, below]{\footnotesize$1-\frac{1}{q}$} node[pos=0.8, below]{\footnotesize$\frac{1}{q}$} (vn);

        \end{scope}
	\end{tikzpicture}
	\caption{Illustration of the edge gadget described in~\cref{lem:edge_gadget} with fractional orientation~$\phi_1$ on the top and~$\phi_2$ on the bottom.
	Vertices~$\{u=v_1, \dots, v_j, \dots, v_n = w\}$ form a path to which edges are added to ensure an indegree of~$\tau'$ w.r.t.~$\phi_1$ for all vertices except~$u$ and~$w$ (which have an indegree of~$\tau' - 1 + \frac{1}{q}$ and~$\tau' - 1$ respectively).
	Vertex~$v_j$ is the first vertex which requires an additional edge to reach~$\tau'$ and~$x$ is the amount of fractional orientation it receives from~$v_{j-1}$.
	To create~$\phi_2$ from~$\phi_1$, fractional orientation is ``pushed'' from~$u$ to~$w$, flipping their indegrees.}
	\label{fig:edge_gadget}
\end{figure}

	Note that since~$m \geq n+1$ (we will next show that the constructed graph has exactly~$m$ edges) and~$n-1$ edges belong to the initial path~$P$, we add at least two edges in that way and hence the edges we try to add are not part of the path already (and we never try to add two edges to the same vertex as~$\tau'<2$).
	Moreover, since~$m \leq 2n-3$ and~$n-1$ edges belong to the initial path, we only add at most~$n-2$ edges and hence there are enough vertices to add sufficiently many edges.
	We next show that after doing the above for all~$2 \leq i \leq n-1$, the graph contains exactly~$m$ edges and~$\orient'_1$ fulfills all requirements from the lemma statement (for target density~$\tau'$).
	By construction, the indegree of~$u$ with respect to~$\orient'_1$ is~$\tau'-1+\frac{1}{q}$ and the indegree of each vertex~$v_i$ with~$2 \leq i \leq n-1$ is~$\tau'$.
	Added over all vertices, the sum of in-degrees is~$(n-1)\tau'-1+\frac{1}{q}$.
	The indegree of~$w$ is some value~$t$ between~$0$ and~$1$ (only the edge~$\{v_{n-1},w\}$ assigns some positive value to~$w$).
	Since the sum of indegrees equals the number of edges, it holds that~$(n-1)\tau' - 1 + \nicefrac{1}{q} + t \in \N$.
	Note that this is the case for~$t = \tau' - 1$ since then, $(n-1)\tau' - 1 + \nicefrac{1}{q} + t = n\tau' - 2 + \nicefrac{1}{q} = \frac{p'n+1}{q}-2$, which is an integer by the choice of~$n$.
	Moreover, $t = \tau' - 1$ is the only value between~$0$ and~$1$ for which this is the case as increasing/decreasing the number of edges by one also increases or decreases~$t$ by~$1$.
	Hence, the number of edges in the constructed graph is~$n \tau' - 2 + \nicefrac{1}{q} = m$, and~$t=\tau'-1$.
	Thus, the orientation~$\orient'_1$ fulfills all requirements of the lemma statement.

	It remains to show the existence of the fractional orientation~$\orient'_2$ and to adapt the construction for target densities larger than~$2$.
	Towards the former, let~$j$ be the index where we added the edge~$\{v_j,w\}$.
	Note that this index exists as~$m \geq n$.
	Intuitively, we modify~$\orient'_1$ by sending an orientation value of~$\nicefrac{1}{q}$ from~$u$ to~$w$ along the path~$(u,v_2, \dots,v_j,w)$.
	Note that by construction all orientations in~$\orient'_1$ are multiples of~$\nicefrac{1}{q}$.
	To construct~$\orient'_2$ from~$\orient'_1$ more formally, we set (for all $i \in [2,j-1]$)
    \begin{align*}
    	\orient'_2 (uv_2)^{u} &\coloneqq \orient'_1(uv_2)^{u} - \frac{1}{q}, &
		\orient'_2 (uv_2)^{v_2} &\coloneqq \orient'_1(uv_2)^{v_2} + \frac{1}{q}, \\
		\orient'_2 (v_iv_{i+1})^{v_i} &\coloneqq \orient'_1(v_iv_{i+1})^{v_i} - \frac{1}{q},&
		\orient'_2 (v_iv_{i+1})^{v_{i+1}} &\coloneqq \orient'_1(v_iv_{i+1})^{v_{i+1}} + \frac{1}{q},\\
		\orient'_2 (v_jw)^{v_j} &\coloneqq \orient'_1(v_jw)^{v_j} - \frac{1}{q}, \text{ and} &
		\orient'_2 (v_jw)^{w} &\coloneqq \orient'_1(v_jw)^{w} + \frac{1}{q}.
    \end{align*}
	Note that if none of the above edges assign a value of~$1$ to the ``later'' vertex (the one with higher index), then the above modification is possible and~$\orient'_2$ also fulfills all requirements from the lemma statement.
	Assume towards a contradiction that some edge~$\{v_i,v_{i'}\}$ assigns a value of~$1$ to the later vertex.
	Let~$i < i'$ without loss of generality and note that by construction~$i' \neq n$ as the edge~$\{v_j,w\}$ is fully assigned to~$v_j$ and~$i \neq 1$ as~$\tau' > 1$ (and hence~$u$ is assigned a positive value~$\tau' - 1 + \nicefrac{1}{q} > \nicefrac{1}{q}$).
	Hence,~$i'=i+1$ and~$i < j$.
	This implies that~$v_i$ is assigned a value of at most~$1< \tau'$ by~$\orient'_1$ as it is only incident to~$\{v_{i-1},v_i\}$ and~$\{v_i,v_{i+1}\}$ as~$i < j$ and~$j$ is the first index where we added an additional edge.
	Thus, the assumption that~$\{v_i,v_{i'}\}$ is fully oriented towards~$v_{i'}$ leads to a contradiction.
	This shows that the above modification is possible and~$\orient'_2$ exists as claimed. %

	Finally, we adapt the construction to increase the density from~$\tau' = (q+r)/q$ by~$c-1$ to~$(q+r)/q + c-1 = (cq + r)/q = p/q = \tau$.
	To this end, we add a balanced graph~$A$ with density~$\tau$ (which always exists and~\mbox{\citet{RV1986}} show how to construct such a graph for each~$\tau \geq 1$).
	Note that since any graph with~$n$ vertices has density at most~$\binom{n}{2}/n=(n-1)/2<n$, it holds that~$A$ contains at least~$c-1$ vertices (but the size of~$A$ is constant for each constant~$\tau$).
	We select a set~$D$ of~$c-1$ arbitrary vertices in~$A$.
	To conclude the construction of the edge gadget, we add an edge between each vertex in~$W'$ and each vertex in~$D$.
	To construct the orientations~$\orient_1$ and~$\orient_2$ for target density~$\tau$, we start with an orientation for~$A$ where each vertex is assigned exactly~$\tau$.
	Note that this exists as~$A$ is balanced and has density~$\tau$.
	For the edges in~$W'$, we use the orientations~$\orient'_1$ for~$\orient_1$ and~$\orient'_2$ for~$\orient_2$, respectively.
	Finally, we orient all edges between~$W'$ and~$D$ completely towards the vertices in~$W'$.
	It is now easy to verify that all requirements of the lemma are satisfied by~$\orient_1$ and~$\orient_2$.
	Clearly if~$p$ and~$q$ are constant, the construction takes constant time.
\end{proof}
}
We are now ready to prove the main result of this section. This will conclude the proof of \cref{thm:dichotomy}, which then follows from \cref{thm:bded_algorithms} (for the algorithmic part) and \cref{thm:hardness_below_one,thm:hardness_between_1_and_2} (for the NP-hardness).

\begin{theorem}
	\label{thm:hardness_between_1_and_2}
	\ac{tbded} is NP-hard for any~$\tau > 1$ with~$\tau \notin \HI$.
\end{theorem}
\begin{proof}
	Let~$\tau = \nicefrac{p}{q} > 1$ where~$p$ and~$q$ are coprime.
	Note that since~$\tau \notin \HI$, it holds that~$q\geq 3$.
	We give a reduction from \acl{vc} on~$q$-regular graphs.
	It is folklore knowledge that this problem is NP-complete.\footnote{Amiri~\cite{Amiri21} showed that the problem cannot be solved in~$2^{o(m)}$ time unless the ETH fails. The same reduction also shows NP-hardness, but this is not explicitly mentioned.}
	The idea is to construct a vertex gadget for each vertex and an edge gadget for each edge in~$G$ using \cref{lem:1-0-edge,lem:edge_gadget}, respectively.

	Let~$(G = (V,E), k)$ be an instance of \acl{vc} where~$G$ is~$q$-regular.
	We construct an equivalent instance~$(G',k')$ of \ac{tbded} as follows.
    For each vertex~$v \in V$, we add a vertex gadget~$H_v$ with a special edge~$e_v$ and an incident vertex~$x_v$ to~$G'$.
	The graph~$H_v$ is any balanced graph with density~$\tau$ (which always exists and~\mbox{\citet{RV1986}} show how to construct such a graph for each~$\tau \geq 1$).
	Let~$\orient_v$ be any fractional orientation for~$H_v$ where each vertex has indegree~$\tau$ and let~$C$ be any cycle in~$H_v$.
	Both can be found in polynomial time.
	Using \cref{lem:1-0-edge}, we construct another fractional orientation~$\orient_v'$ where each vertex has indegree~$\tau$ and one edge~$e$ fully assigns its orientation to one of its endpoints~$u$.
	We set~$e_v \coloneqq e$ and~$x_v \coloneqq u$.
	This concludes the construction of the vertex gadget.
	For each edge~$e=\{u,v\}\in E$, we add the edge gadget~$W_e$, which is the graph obtained from the construction in \cref{lem:edge_gadget} for density~$\tau$ with special vertices~$w_u$ and~$w_v$.
	We connect the gadgets by adding the edges~$\{x_u, w_u\}$ and~$\{x_v,w_v\}$ to our construction.
	Note that we can use \cref{lem:edge_gadget} since we assumed that~$p$ and~$q$ are coprime.
	To conclude the construction, we set~$k'\coloneqq k$.
	A visualization of the reduction is shown in \cref{fig:hardness_above_one}.
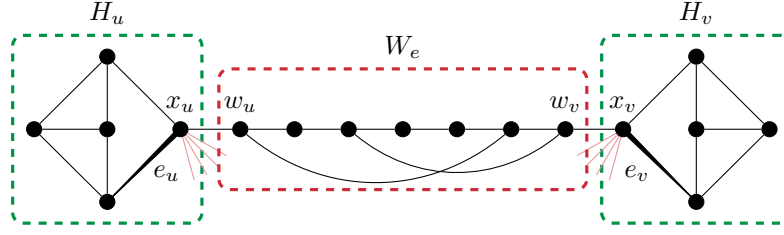
\begin{figure}
	\centering
	\begin{tikzpicture}[node distance = .75cm]

        \begin{scope}
         \node[vertex_small] (lc){};
         \node[vertex_small,above=of lc] (l1){};
         \node[vertex_small,left=of lc] (l2){};
         \node[vertex_small,below=of lc] (l3){};
         \node[vertex_small,right=of lc,label={[label position=above]{$x_u$}}] (l4){};

         \node[inner sep = 5pt, very thick, italyGreen, draw,rounded corners, dashed, fit=(l1) (l2) (l3) (l4),label=above:$H_u$](C) {};

         \foreach \i/\j in {l1/l2, l2/l3, l4/l1, lc/l1, lc/l2, lc/l3}
         {
            \draw (\i) -- (\j);
         }

         \draw[fill] (l3.center) -- ($(l4.center)!0.3!(l4.north west)$) -- ($(l4.center)!0.3!(l4.south east)$) -- (l3.center);
         \draw (l3) -- node[midway,xshift=.3cm,yshift=-.1cm]{$e_u$} (l4);

         \foreach \i in {1,...,4}
         {
            \draw[italyRed!50] (l4) -- ++(\i*-15-15:.7);
         }

        \end{scope}

        \begin{scope}[xshift=7.75cm]
         \node[vertex_small] (rc){};
         \node[vertex_small,above=of rc] (r1){};
         \node[vertex_small,right=of rc] (r2){};
         \node[vertex_small,below=of rc] (r3){};
         \node[vertex_small,left=of rc,label={[label position=above]{$x_v$}}] (r4){};

         \node[inner sep = 5pt, very thick, italyGreen, draw,rounded corners, dashed, fit=(r1) (r2) (r3) (r4),label=above:$H_v$](C) {};

         \foreach \i/\j in {r1/r2, r2/r3, r4/r1, rc/r1, rc/r2, rc/r3}
         {
            \draw (\i) -- (\j);
         }

         \draw[fill] (r3.center) -- ($(r4.center)!0.3!(r4.north east)$) -- ($(r4.center)!0.3!(r4.south west)$) -- (r3.center);
         \draw (r3) -- node[midway,xshift=-.3cm,yshift=-.1cm]{$e_v$} (r4);

         \foreach \i in {1,...,4}
         {
            \draw[italyRed!50] (r4) -- ++(\i*15+195:.7);
         }

        \end{scope}

        \begin{scope}[decoration={markings, mark=at position 0.2 with{\arrowreversed{Stealth}}, mark=at position 0.8 with{\arrow{Stealth}}}, start chain=going right, node distance=.5cm,xshift=1.75cm]
            \node[on chain, vertex_small, label={[label position=above]{$w_u$}}] (v1){};
            \node[on chain, vertex_small] (v2){};
            \node[on chain, vertex_small] (v3){};
            \node[on chain, vertex_small] (v4){};
            \node[on chain, vertex_small] (v5){};
            \node[on chain, vertex_small] (v6){};
            \node[on chain, vertex_small, label={[label position=above]{$w_v$}}] (v7){};

            \node[inner sep = 5pt, minimum height=1.6cm, very thick, italyRed, draw,rounded corners, dashed, fit=(v1) (v7),label=above:$W_e$](C) {};

            \foreach \i/\j in {v1/v2, v2/v3, v3/v4, v4/v5, v5/v6, v6/v7}
            {
                \draw (\i) -- (\j);
            }

            \draw[bend right=40] (v1) to (v6);
            \draw[bend right=40] (v3) to (v7);

        \end{scope}

        \draw (v1)--(l4);
        \draw (v7)--(r4);
	\end{tikzpicture}
	\caption{Illustration of the hardness reduction for~$\tau =p/q > 1$.
	Depicted is the construction for two vertices~$u,v$ in the \textsc{Vertex Cover} instance connected by an edge~$e=\{u,v\}$ with~$\tau=\frac{7}{5}$.
	To the left and right are vertex gadgets~$H_u$ and~$H_v$ (highlighted in green) with special vertices~$x_u$ and~$x_v$ and the special edges~$e_u$ and~$e_v$ oriented fully towards~$x_u$ and~$x_v$ respectively (indicated with the thickening edge).
	The vertex gadgets are connected by the edge gadget~$W_e$ with special vertices~$w_u$ and~$w_v$.
	Light red lines indicate edges from~$x_u$ and~$x_v$ to other edge gadgets.
	Note that for $\tau=\frac{7}{5}$, a balanced graph on~5 vertices with~7 edges can be created using the construction of \citet{RV1986}.
	For~$p > {q \choose 2}$, observe that a graph on~$pq$ vertices and~$p^2$ edges (of density~$\tau =p/q$) can be created.}
	\label{fig:hardness_above_one}
\end{figure}

	Since the reduction can be computed in polynomial time as the graphs~$H_v$ and~$W_e$ have constant size, it only remains to show correctness.
	To this end, first assume that~$(G,k)$ is a yes-instance of \acl{vc} and let~$C$ be a vertex cover of size at most~$k$ in~$G$.
	For each~$v \in C$, we add the edge~$e_v$ in the vertex gadget~$H_v$ to the solution set~$F$.
	Note that these are~$|C| \leq k = k'$ edges.
	By \cref{lem:density_bound}, it suffices to provide a fractional orientation~$\orient$ of~$G' - F$ with~$\Delta_{\orient}^- \leq \tau$, which we construct as follows.
	We use the fractional orientation~$\orient'_v$ for each vertex gadget (with potentially the edge~$e_v$ removed).
	Hence, the indegree of each vertex in a vertex gadget is~$\tau$ unless it is the vertex~$x_v$ in a vertex gadget corresponding to a vertex~$v \in C$.
	In this case, the indegree is~$\tau-1$.
	We next construct the fractional orientation for each edge gadget~$W_e$ for an edge~$e = \{u,v\}$ (including the two edges connecting it to vertex gadgets).
	We assume without loss of generality that~$u \in C$ (if both are contained, we pick one endpoint arbitrarily; the case where only~$v$ is part of the vertex cover is analogous).
	This means that the vertex~$x_u$ has indegree~$\tau-1$ in the current orientation.
	We use \cref{lem:edge_gadget} to obtain a fractional orientation of the edges in~$W_e$ such that now all vertices in it except for~$w_u$ and~$w_v$ have indegree~$\tau$ and~$w_u$ has indegree~$\tau-1+\nicefrac{1}{q}$ and~$w_v$ has indegree~$\tau-1$.
	We orient the edge~$\{w_v, x_v\}$ completely towards~$w_v$ and the edge~$\{w_u, x_u\}$ with a value of~$1-\nicefrac{1}{q}$ to~$w_u$ and with a value of~$\nicefrac{1}{q}$ to~$x_u$.
	By doing this for every edge~$e$, we obtain a fractional orientation~$\phi$ of~$G' - F$ where each vertex has indegree at most~$\tau$ since each vertex~$x_u$ is only connected to at most~$q$ edge gadgets that assign an orientation of~$\nicefrac{1}{q}$ each.

	Conversely, assume that the constructed instance of \ac{tbded} is a yes-instance and let~$F$ be a solution of size at most~$k'=k$.
	We construct a vertex cover $C$ of size at most~$k$ for $G$ as follows.
	For all edges $e \in F$ which are contained in some vertex gadget~$H_v$ or connect a vertex in~$H_v$ with an edge gadget, we add~$v$ to the vertex cover~$C$.
	If~$e\in F$ is contained in an edge gadget~$W_{e'}$, then we arbitrarily pick one of the two endpoints of~$e'$ and add it to~$C$.
	Since we selected at most~$k$ vertices, it only remains to prove that~$C$ is a vertex cover.
	Assume towards a contradiction that this is not the case.
	Then, there exists some edge~$e=\{u,v\}$ where~$u,v \notin C$.
	Note that this implies that no edge incident to a vertex in the vertex gadgets corresponding to~$u$ and~$v$ and no edge in the edge gadget corresponding to~$e$ are contained in~$F$.
	Let~$n_v$ be the number of vertices in a vertex gadget and let~$n_e$ be the number of vertices in an edge gadget.
	The union of the three aforementioned gadgets (plus the two edges between the gadgets) contains~$2n_v+n_e$ vertices and~$\tau(2n_v) + (\tau n_e - 2 + \nicefrac{1}{q}) + 2 = \tau (2n_v + n_e) + \nicefrac{1}{q} > \tau (2n_v + n_e)$ edges.
	Hence, the density of this subgraph is larger than~$\tau$, contradicting the fact that~$F$ is a solution.
	This concludes the proof.
\end{proof}

\section{Conclusion}
\label{sec:concl}

We provided a complexity dichotomy of \ac{tbded} with respect to the target density~$\tau$ by showing that the problem is polynomial-time solvable if and only if~$\tau$ is half-integral or smaller than~$\nicefrac{2}{3}$ (assuming P${}\ne{}$NP).
We also observed an encouraging result regarding the treewidth of the input graph.
Studying \ac{tbded} with respect to further natural parameters like clique-width, or the diameter of the input graph and/or from the viewpoint of approximation algorithms are natural next steps.
Another next step is to analyze cases in which~$\tau$ is not constant but can depend on e.g.~$n$. 
Our polynomial-time algorithms for (half-)integral target densities generalize to this setting.
However, our hardness results and the FPT algorithm for treewidth do not.
It is known that for some value~$\tau = 1 - 1/o(n)$ depending on~$n$ the problem becomes W[1]-hard with respect to the treewidth \cite{BNV25} and it is polynomial-time solvable for~$\tau = 1-\nicefrac{1}{n}$ since this is equivalent to computing a feedback edge set of the input graph.
Another generalization our polynomial-time algorithms cover are weighted problem variants of \ac{tbded} as the weighted variants of \ac{flow} and \ac{gf} are polynomial-time solvable~\cite{DP18,Sch03}.

Concerning \ac{rflow}, an interesting question is whether the efficiently solvable cases discovered in this work (max-gap at most one, bounded treewidth) can be extended to cover larger classes of networks.
For example, the NP-hardness for the case that only some arcs have an even flow restriction requires to ``stack'' several of them on paths:
Each path from~$s$ to~$t$ in the constructed acyclic flow instance contains exactly three arcs requiring even flow (see \cref{app:np_hardness_subset_even_flow}).
Hence, the question arises if the problem is polynomial-time solvable if each such restricted arc is incident to~$s$ or~$t$.

\nocite{MS12}
\nocite{DP18}
\nocite{MSS21}
\nocite{CCK25}

\newpage

\bibliographystyle{plainnat}
\bibliography{References.bib}

\newpage

\appendix

\section{NP-hardness of \textsc{Subset Even Flow}}

\label{app:np_hardness_subset_even_flow}
Consider the \acl{flow} problem on digraph~$D$ with integral capacities~$c\colon E \rightarrow \N$ with the additional constraint that for a proper subsetof arcs~$F \subsetneq A(D)$ the capacities must be even.
This variant is called \textsc{Subset Even Flow} and it is NP-hard.
While we were unable to find the result in the scientific literature, it is known and discussed publicly. We now repeat the proof sketch given by Neal Young on stackexchange\footnote{\url{https://cstheory.stackexchange.com/questions/51243/maximum-flow-with-parity-requirement} \url{-on-certain-edges}}.

\begin{theorem}[Neal Young]
	\textsc{Subset Even Flow} is NP-hard, even if all capacities are at most two.
\end{theorem}

\begin{proof}
	We reduce from \textsc{Independent Set} on cubic graphs (each vertex has exactly three neighbors), which is known to be NP-hard~\cite{GJ79}.
	Given a cubic graph~$G$, produce a \textsc{Subset Even Flow} instance~$(D, \capacity, F)$ as follows:
	Start with only a source~$s$ and a sink~$t$ in~$D$.
	Then, for every vertex~$v \in V(G)$ with neighbors~$u_1, u_2, u_3$ add a vertex gadget~$g_v$ to the instance (see the left in~\cref{fig:subset_even_flow} for an illustration), which consists of six vertices~$v_1,v_2,v_3$ (forming the bottom layer in \cref{fig:subset_even_flow}) and $vu_1,vu_2,vu_3$ (the top layer in the figure), and the arcs~$(s,v_i)$ with even capacity~$\capacity(s,v_i) \in \{0,2\}$ for~$i \in [3]$ (that is, add~$(s,v_i)$ to~$F$) and~$(v_i,vu_j)$ with capacity~$\capacity(v_i,vu_j) = 1$ for~$i,j \in [3], i \ne j$.
	
	For every edge~$\{u,w\} \in E(G)$ choose the unique vertices~$uw$, $wu$ in the top layer of~$g_u$ and~$g_w$ and build a gadget~$g_{uw}$ as shown in the right of~\cref{fig:subset_even_flow}:
	Add a vertex~$e_{uw}$ and the arcs~$(uw,e_{uw})$, $(wu,e_{uw})$, and~$(e_{uw},t)$; all three arcs have even capacities requirement.
	This completes the construction.

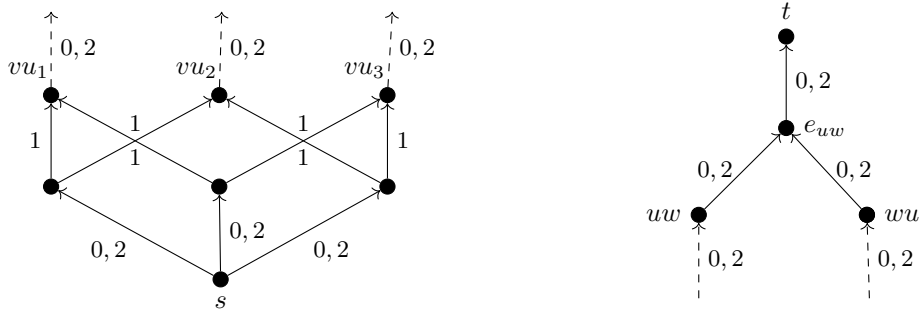
\begin{figure}
	\centering

	\begin{minipage}{0.4\textwidth}
	 \begin{tikzpicture}
			\begin{scope}
				\begin{scope}[node distance=1cm and 2cm]
					\node[] (v11) {};
					\node[vertex_small, right=of v11,label=below:$s$] (v12) {};
					\node[right=of v12] (v13) {};
					\node[vertex_small, above=of v11] (v21) {};
					\node[vertex_small, right=of v21] (v22) {};
					\node[vertex_small, right=of v22] (v23) {};
					\node[vertex_small, above=of v21, label={[xshift=-.3cm]$vu_1$}] (v31) {};
					\node[vertex_small, right=of v31, label={[xshift=-.3cm]$vu_2$}] (v32) {};
					\node[vertex_small, right=of v32, label={[xshift=-.3cm]$vu_3$}] (v33) {};
					\node[above=of v31] (v41) {};
					\node[right= of v41] (v42) {};
					\node[right= of v42] (v43) {};
				\end{scope}

				\foreach \i/\j/\txt/\p in {%
					v12/v21/${0,2}$/below left,
					v12/v22/${0,2}$/right,
					v12/v23/${0,2}$/below right,
					v21/v31/$1$/left,
					v21/v32/$1$/above,
					v22/v31/$1$/below,
					v22/v33/$1$/above,
					v23/v32/$1$/below,
					v23/v33/$1$/right}
					\draw [->] (\i) -- node[font=\small,\p] {\txt} (\j);

                \foreach \i/\j/\txt/\p in {%
					v31/v41/${0,2}$/right,
					v32/v42/${0,2}$/right,
					v33/v43/${0,2}$/right}
					\draw [dashed, ->] (\i) -- node[font=\small,\p] {\txt} (\j);
			\end{scope}

		\end{tikzpicture}
	\end{minipage}
	\hfill
	\begin{minipage}{0.4\textwidth}
	 \begin{tikzpicture}
			\begin{scope}
				\begin{scope}[node distance=1cm and 2cm]
					\node[] (v11) {};
					\node[right=of v11] (v12) {};
					\node[vertex_small, above=of v11, label=left:$uw$] (v21) {};
					\node[vertex_small, right=of v21, label=right:$wu$] (v22) {};
					\node[vertex_small, right=of v21] (v22) {};
					\node[vertex_small, above right=1cm and 1cm of v21,label=right:$e_{uw}$] (v3) {};
					\node[vertex_small, above=of v3, label=above:$t$] (v4) {};
				\end{scope}

				\foreach \i/\j/\txt/\p in {%
					v21/v3/${0,2}$/left,
					v22/v3/${0,2}$/right,
					v3/v4/${0,2}$/right}
					\draw [->] (\i) -- node[font=\small,\p] {\txt} (\j);

                \foreach \i/\j/\txt/\p in {%
					v11/v21/${0,2}$/right,
					v12/v22/${0,2}$/right}
					\draw [dashed, ->] (\i) -- node[font=\small,\p] {\txt} (\j);
			\end{scope}

		\end{tikzpicture}
	\end{minipage}

	\caption{On the left side a vertex gadget~$g_v$ for a vertex~$v \in V(G)$ with neighbors~$u_1,u_2,u_3$ is shown.
	The vertex~$s$ at the bottom is sole the source in the flow network and appears in every vertex gadget.
	On the right side an edge gadget~$g_{uw}$ for an edge~$\{u,v\} \in E(G)$ is depicted.
	Each bottom vertex of the edge gadget refers to a vertex in the top layer of vertex gadgets~$g_u$ and~$g_v$.
	The top vertex~$t$ is the sole sink in the network and appears in every edge gadget.
	Dashed lines indicate how much flow can enter or leave a vertex when properly connected, but they are not part of the construction.
	Capacity~${0, 2}$ indicates that only a flow of~0 or~2 (but not~1 or~1.5) can be sent along an arc.
	}
	\label{fig:subset_even_flow}
\end{figure}

	To prove correctness, we show that~$G$ contains an independent set of size~$k$ if and only if~$D$ admits an $s$-$t$-flow of value~$6k$ respecting the capacity restriction.
	This is based on the observation that the edge gadget encodes a NAND operation, allowing only one of the two vertices~$uw$ (from~$g_u$) and~$wu$ (from~$g_w$) to send flow through~$e_{uw}$ to~$t$.
	
	Given an independent set~$S \subseteq V(G)$ of size~$k$, send flow through all edges in each vertex gadget~$g_v$ with~$v \in S$.
	As~$S$ is an independent set, the flow can pass through the edge gadget~$g_{vu_i}$ to~$t$.
	Thus there is a flow of~$6k$ in~$D$.
	
	Conversely, observe that each vertex~$vu_i$ in~$D$ has one outgoing arc with even flow requirement.
	Thus, a flow solution respecting the parity restriction either saturates all edges in~$g_v$ or assigns flow~0 to all of them.
	Hence, a flow has value~$6k$ for some~$k \in \N$ and corresponds to an independent set of size~$k$ in~$G$.
\end{proof}

\end{document}